\documentclass{amsart}
\usepackage{amsfonts}
\usepackage{amsmath,amscd}
\usepackage{amsthm}
\usepackage{amssymb}
\usepackage{latexsym}
\usepackage{ulem}   

\numberwithin{equation}{section}

\usepackage{color}

\newtheorem{thm}{Theorem}[section]

\newtheorem{lem}[thm]{Lemma}

\newtheorem{prop}[thm]{Proposition}

\newtheorem{defn}[thm]{Definition}

\newtheorem{rem}[thm]{Remark}

\numberwithin{equation}{section}
\begin{document}
\newcommand{\beqa}{\begin{eqnarray}}
\newcommand{\eeqa}{\end{eqnarray}}
\newcommand{\thmref}[1]{Theorem~\ref{#1}}
\newcommand{\secref}[1]{Sect.~\ref{#1}}
\newcommand{\lemref}[1]{Lemma~\ref{#1}}
\newcommand{\propref}[1]{Proposition~\ref{#1}}
\newcommand{\corref}[1]{Corollary~\ref{#1}}
\newcommand{\remref}[1]{Remark~\ref{#1}}
\newcommand{\er}[1]{(\ref{#1})}
\newcommand{\nc}{\newcommand}
\newcommand{\rnc}{\renewcommand}

\nc{\cal}{\mathcal}

\nc{\goth}{\mathfrak}
\rnc{\bold}{\mathbf}
\renewcommand{\frak}{\mathfrak}
\renewcommand{\Bbb}{\mathbb}

\nc{\trr}{\triangleright}
\nc{\trl}{\triangleleft}

\newcommand{\id}{\text{id}}
\nc{\Cal}{\mathcal}
\nc{\Xp}[1]{X^+(#1)}
\nc{\Xm}[1]{X^-(#1)}
\nc{\on}{\operatorname}
\nc{\ch}{\mbox{ch}}
\nc{\Z}{{\bold Z}}
\nc{\J}{{\mathcal J}}
\nc{\C}{{\bold C}}
\nc{\Q}{{\bold Q}}
\nc{\oC}{{\widetilde{C}}}
\nc{\oc}{{\tilde{c}}}
\nc{\ocI}{ \overline{\cal I}}
\nc{\og}{{\tilde{\gamma}}}
\nc{\lC}{{\overline{C}}}
\nc{\lc}{{\overline{c}}}
\nc{\Rt}{{\tilde{R}}}

\nc{\tW}{{\textsf{W}}}
\nc{\tG}{{\textsf{G}}}

\nc{\cL}{{\cal{L}}}
\nc{\cK}{{\cal{K}}}

\nc{\tw}{{\textsf{w}}}
\nc{\tg}{{\textsf{g}}}

\nc{\tx}{{\textsf{x}}}
\nc{\tho}{{\textsf{h}}}
\nc{\tk}{{\textsf{k}}}
\nc{\tep}{{\bf{\cal E}}}

\nc{\te}{{\textsf{e}}}
\nc{\tf}{{\textsf{f}}}
\nc{\tK}{{\textsf{K}}}

\nc{\odel}{{\overline{\delta}}}

\def\pr#1{\left(#1\right)_\infty}  

\renewcommand{\P}{{\mathcal P}}
\nc{\N}{{\Bbb N}}
\nc\beq{\begin{equation}}
\nc\enq{\end{equation}}
\nc\lan{\langle}
\nc\ran{\rangle}
\nc\bsl{\backslash}
\nc\mto{\mapsto}
\nc\lra{\leftrightarrow}
\nc\hra{\hookrightarrow}
\nc\sm{\smallmatrix}
\nc\esm{\endsmallmatrix}
\nc\sub{\subset}
\nc\ti{\tilde}
\nc\nl{\newline}
\nc\fra{\frac}
\nc\und{\underline}
\nc\ov{\overline}
\nc\ot{\otimes}

\nc\ochi{\overline{\chi}}
\nc\bbq{\bar{\bq}_l}
\nc\bcc{\thickfracwithdelims[]\thickness0}
\nc\ad{\text{\rm ad}}
\nc\Ad{\text{\rm Ad}}
\nc\Hom{\text{\rm Hom}}
\nc\End{\text{\rm End}}
\nc\Ind{\text{\rm Ind}}
\nc\Res{\text{\rm Res}}
\nc\Ker{\text{\rm Ker}}
\rnc\Im{\text{Im}}
\nc\sgn{\text{\rm sgn}}
\nc\tr{\text{\rm tr}}
\nc\Tr{\text{\rm Tr}}
\nc\supp{\text{\rm supp}}
\nc\card{\text{\rm card}}
\nc\bst{{}^\bigstar\!}
\nc\he{\heartsuit}
\nc\clu{\clubsuit}
\nc\spa{\spadesuit}
\nc\di{\diamond}
\nc\cW{\cal W}
\nc\cG{\cal G}
\nc\cZ{\cal Z}
\nc\ocW{\overline{\cal W}}
\nc\ocZ{\overline{\cal Z}}
\nc\al{\alpha}
\nc\bet{\beta}
\nc\ga{\gamma}
\nc\de{\delta}
\nc\ep{\epsilon}
\nc\io{\iota}
\nc\om{\omega}
\nc\si{\sigma}
\rnc\th{\theta}
\nc\ka{\kappa}
\nc\la{\lambda}
\nc\ze{\zeta}

\nc\vp{\varpi}
\nc\vt{\vartheta}
\nc\vr{\varrho}

\nc\odelta{\overline{\delta}}
\nc\Ga{\Gamma}
\nc\De{\Delta}
\nc\Om{\Omega}
\nc\Si{\Sigma}
\nc\Th{\Theta}
\nc\La{\Lambda}

\nc\boa{\bold a}
\nc\bob{\bold b}
\nc\boc{\bold c}
\nc\bod{\bold d}
\nc\boe{\bold e}
\nc\bof{\bold f}
\nc\bog{\bold g}
\nc\boh{\bold h}
\nc\boi{\bold i}
\nc\boj{\bold j}
\nc\bok{\bold k}
\nc\bol{\bold l}
\nc\bom{\bold m}
\nc\bon{\bold n}
\nc\boo{\bold o}
\nc\bop{\bold p}
\nc\boq{\bold q}
\nc\bor{\bold r}
\nc\bos{\bold s}
\nc\bou{\bold u}
\nc\bov{\bold v}
\nc\bow{\bold w}
\nc\boz{\bold z}

\nc\ba{\bold A}
\nc\bb{\bold B}
\nc\bc{\bold C}
\nc\bd{\bold D}
\nc\be{\bold E}
\nc\bg{\bold G}
\nc\bh{\bold H}
\nc\bi{\bold I}
\nc\bj{\bold J}
\nc\bk{\bold K}
\nc\bl{\bold L}
\nc\bm{\bold M}
\nc\bn{\bold N}
\nc\bo{\bold O}
\nc\bp{\bold P}
\nc\bq{\bold Q}
\nc\br{\bold R}
\nc\bs{\bold S}
\nc\bt{\bold T}
\nc\bu{\bold U}
\nc\bv{\bold V}
\nc\bw{\bold W}
\nc\bz{\bold Z}
\nc\bx{\bold X}

\nc\ca{\mathcal A}
\nc\cb{\mathcal B}
\nc\cc{\mathcal C}
\nc\cd{\mathcal D}
\nc\ce{\mathcal E}
\nc\cf{\mathcal F}
\nc\cg{\mathcal G}
\rnc\ch{\mathcal H}
\nc\ci{\mathcal I}
\nc\cj{\mathcal J}
\nc\ck{\mathcal K}
\nc\cl{\mathcal L}
\nc\cm{\mathcal M}
\nc\cn{\mathcal N}
\nc\co{\mathcal O}
\nc\cp{\mathcal P}
\nc\cq{\mathcal Q}
\nc\car{\mathcal R}
\nc\cs{\mathcal S}
\nc\ct{\mathcal T}
\nc\cu{\mathcal U}
\nc\cv{\mathcal V}
\nc\cz{\mathcal Z}
\nc\cx{\mathcal X}
\nc\cy{\mathcal Y}

\nc\e[1]{E_{#1}}
\nc\ei[1]{E_{\delta - \alpha_{#1}}}
\nc\esi[1]{E_{s \delta - \alpha_{#1}}}
\nc\eri[1]{E_{r \delta - \alpha_{#1}}}
\nc\ed[2][]{E_{#1 \delta,#2}}
\nc\ekd[1]{E_{k \delta,#1}}
\nc\emd[1]{E_{m \delta,#1}}
\nc\erd[1]{E_{r \delta,#1}}

\nc\ef[1]{F_{#1}}
\nc\efi[1]{F_{\delta - \alpha_{#1}}}
\nc\efsi[1]{F_{s \delta - \alpha_{#1}}}
\nc\efri[1]{F_{r \delta - \alpha_{#1}}}
\nc\efd[2][]{F_{#1 \delta,#2}}
\nc\efkd[1]{F_{k \delta,#1}}
\nc\efmd[1]{F_{m \delta,#1}}
\nc\efrd[1]{F_{r \delta,#1}}

\nc\fa{\frak a}
\nc\fb{\frak b}
\nc\fc{\frak c}
\nc\fd{\frak d}
\nc\fe{\frak e}
\nc\ff{\frak f}
\nc\fg{\frak g}
\nc\fh{\frak h}
\nc\fj{\frak j}
\nc\fk{\frak k}
\nc\fl{\frak l}
\nc\fm{\frak m}
\nc\fn{\frak n}
\nc\fo{\frak o}
\nc\fp{\frak p}
\nc\fq{\frak q}
\nc\fr{\frak r}
\nc\fs{\frak s}
\nc\ft{\frak t}
\nc\fu{\frak u}
\nc\fv{\frak v}
\nc\fz{\frak z}
\nc\fx{\frak x}
\nc\fy{\frak y}

\nc\fA{\frak A}
\nc\fB{\frak B}
\nc\fC{\frak C}
\nc\fD{\frak D}
\nc\fE{\frak E}
\nc\fF{\frak F}
\nc\fG{\frak G}
\nc\fH{\frak H}
\nc\fJ{\frak J}
\nc\fK{\frak K}
\nc\fL{\frak L}
\nc\fM{\frak M}
\nc\fN{\frak N}
\nc\fO{\frak O}
\nc\fP{\frak P}
\nc\fQ{\frak Q}
\nc\fR{\frak R}
\nc\fS{\frak S}
\nc\fT{\frak T}
\nc\fU{\frak U}
\nc\fV{\frak V}
\nc\fZ{\frak Z}
\nc\fX{\frak X}
\nc\fY{\frak Y}
\nc\tfi{\ti{\Phi}}
\nc\bF{\bold F}
\rnc\bol{\bold 1}

\nc\ua{\bold U_\A}

\nc\qinti[1]{[#1]_i}
\nc\q[1]{[#1]_q}
\nc\xpm[2]{E_{#2 \delta \pm \alpha_#1}}  
\nc\xmp[2]{E_{#2 \delta \mp \alpha_#1}}
\nc\xp[2]{E_{#2 \delta + \alpha_{#1}}}
\nc\xm[2]{E_{#2 \delta - \alpha_{#1}}}
\nc\hik{\ed{k}{i}}
\nc\hjl{\ed{l}{j}}
\nc\qcoeff[3]{\left[ \begin{smallmatrix} {#1}& \\ {#2}& \end{smallmatrix}
\negthickspace \right]_{#3}}
\nc\qi{q}
\nc\qj{q}

\nc\ufdm{{_\ca\bu}_{\rm fd}^{\le 0}}


\nc\isom{\cong} 

\nc{\pone}{{\Bbb C}{\Bbb P}^1}
\nc{\pa}{\partial}
\def\H{\mathcal H}
\def\L{\mathcal L}
\nc{\F}{{\mathcal F}}
\nc{\Sym}{{\goth S}}
\nc{\A}{{\mathcal A}}
\nc{\arr}{\rightarrow}
\nc{\larr}{\longrightarrow}

\nc{\ri}{\rangle}
\nc{\lef}{\langle}
\nc{\W}{{\mathcal W}}
\nc{\uqatwoatone}{{U_{q,1}}(\su)}
\nc{\uqtwo}{U_q(\goth{sl}_2)}
\nc{\dij}{\delta_{ij}}
\nc{\divei}{E_{\alpha_i}^{(n)}}
\nc{\divfi}{F_{\alpha_i}^{(n)}}
\nc{\Lzero}{\Lambda_0}
\nc{\Lone}{\Lambda_1}
\nc{\ve}{\varepsilon}
\nc{\bepsilon}{\bar{\epsilon}}
\nc{\bak}{\bar{k}}
\nc{\phioneminusi}{\Phi^{(1-i,i)}}
\nc{\phioneminusistar}{\Phi^{* (1-i,i)}}
\nc{\phii}{\Phi^{(i,1-i)}}
\nc{\Li}{\Lambda_i}
\nc{\Loneminusi}{\Lambda_{1-i}}
\nc{\vtimesz}{v_\ve \otimes z^m}

\nc{\asltwo}{\widehat{\goth{sl}_2}}
\nc\ag{\widehat{\goth{g}}}  
\nc\teb{\tilde E_\boc}
\nc\tebp{\tilde E_{\boc'}}

\newcommand{\LR}{\bar{R}}
\newcommand{\eeq}{\end{equation}}
\newcommand{\ben}{\begin{eqnarray}}
\newcommand{\een}{\end{eqnarray}}

\title[Freidel-Maillet presentations of $U_q(sl_2)$]{Freidel-Maillet type presentations of $U_q(sl_2)$} 
\author{Pascal Baseilhac}
\address{Institut Denis-Poisson CNRS/UMR 7013 - Universit\'e de Tours - Universit\'e d'Orl\'eans
Parc de Grammont, 37200 Tours, 
FRANCE}
\email{pascal.baseilhac@idpoisson.fr}

\begin{abstract} A unified framework for the Chevalley and equitable presentation of $U_q(sl_2)$ is introduced. It is given in terms of a system of Freidel-Maillet type equations satisfied by a pair of quantum K-operators $\cK^\pm$, whose entries are expressed in terms of either Chevalley or equitable generators. The Hopf algebra structure is reconsidered in light of this unified framework, and interwining relations for each pair of $\cK^\pm$ are obtained. A K-operator solving a spectral parameter dependent Freidel-Maillet type equation is also considered. Different specializations of this K-operator are shown to admit a decomposition in terms of $\cK^\pm$ of Chevalley or equitable type. Explicit examples of K-matrices without/with spectral parameter are derived by specializing the K-operators previously obtained.     
\end{abstract}

\maketitle

\vskip -0.5cm


\vspace{5mm}

\section{Introduction}
The Faddeev-Reshetikhin-Takhtajan (FRT) presentation of quantum algebras \cite{FRT87,FRT89}  has played a central role in the development of integrable quantum field theory and statistical mechanics, providing a powerful framework for studying the properties of the system under consideration (spectrum, scattering data, correlation functions,...).  For any quantum algebra $U_q(g)$, given a R-matrix satisfying
the Yang-Baxter equation
\begin{align}
R_{12}R_{13}R_{23}=R_{23}R_{13}R_{12}\ ,\label{YB0}
\end{align}
where the standard notation $R_{ij}\in \mathrm{End}({\cal V}_i\otimes {\cal V}_j)$ is used, a FRT type presention  
 gives a matrix realization  of $U_q(g)$ of the following form. It consists of a pair of quantum L-operators with entries $(\cL^\pm_{i,j})_{1\leq i,j \leq n}\in U_q(g)$ obeying a system of exchange relations called the Yang-Baxter algebra. The defining relations can be written as 
%
\beqa
\cL^\pm_{i,i}\cL^\mp_{i,i}=1 \ ,\qquad 
R\cL_1^\pm \cL_2^\pm = \cL_2^\pm \cL_1^\pm R \ ,\qquad 
R\cL_1^+ \cL_2^- = \cL_2^- \cL_1^+ R \ ,\label{YBAint}
\eeqa
where the shorthand notation $\cL_1=\cL \otimes {\mathbb I}$, $\cL_2= {\mathbb I} \otimes \cL$ is used. The extension of (\ref{YBAint}) to any affine Kac-Moody algebra $U_q(\widehat{g})$ has been extensively studied, starting from \cite{RS,DF93}, until the recent achievements \cite{JLM,JLMBD,LP21}. In this case, the corresponding Yang-Baxter equation (\ref{YB0}) depends on an indeterminate  $u$, called the spectral parameter in the mathematical physics literature. In this latter context,  corresponding L-operators and R-matrix solutions are the key ingredients for the  construction of transfer matrices that provide generating functions for mutually commuting quantities in quantum integrable systems. Furthermore, powerful standard techniques such as the Bethe ansatz and separation of variables essentially depend on the presentation (\ref{YBAint}).
 For a recent review, see e.g. \cite{Slav18}. Solutions of (\ref{YB0}) also arise in the context of knot theory and quotients of the braid group, see e.g. \cite[Chapter 15]{CPb}.\vspace{1mm} 

In all these works, it is important to point out that besides the presentation (\ref{YBAint}), the quantum algebra $U_q(g)$ always arises in the so-called {\it Chevalley} (or Drinfeld-Jimbo) presentation with generators $\{E_i,F_i,K^{\pm 1}_i,i=1,...,n\}$ \cite{Jim86,Dr0}. However, in recent years
a different presentation for $U_q(g)$ has been introduced, called {\it equitable}, with generators $\{X^{\pm 1}_i,Y_i,Z_i,i=1,...,n\}$ \cite{ITW,Ter05}. Since \cite{ITW}, the equitable presentation of $U_q(sl_2)$  has been studied from various perspectives: 
  algebra \cite{Tereq11,Tereq15b,GVZeq15,FJeq16,LXJeq19}, combinatorics  and representation theory \cite{Aeq11,MReq12,Feq13,HGeq13,IRTeq13,BCTeq14,Tereq14,Yeq15,Neq15,Tereq15a,SGeq16,Tereq20} for instance. Also, note that irreducible finite dimensional representations have been studied in details in \cite{ITW,Ter09}.	In addition, let us mention that the subalgebra generated by $\{Y,Z\}$ of the equitable $U_q(sl_2)$ is isomorphic to the non-homogeneous Borel subalgebra considered in \cite{LV17,Vo,LV19}.  Importantly, contrary to the Chevalley presentation which fits with the FRT framework, up to now a similar set up for the equitable presentation has remained unknown.\vspace{1mm} 

In this paper, a unified framework for the Chevalley and equitable presentation of $U_q(sl_2)$ is proposed. The basic data consists of two R-matrices $R$, $R^{(0)}$, and a pair of K-operators with $(\cK^\pm_{i,j})_{1\leq i,j \leq 2}\in U_q(sl_2)$  - written in terms of the Chevalley or equitable generators - satisfying a system of Freidel-Maillet type equations\footnote{See also \cite{NC92,Bab92,KS}. Note that equations of the form (\ref{FMint}) arise naturally in the context of braided Yang-Baxter algebras. See e.g. \cite{Parm}.} \cite{FM91}.
In both Chevalley and equitable cases, the defining relations of $U_q(sl_2)$  can be written in the form (see Theorem \ref{thm:isoUslFM}):
\beqa
(\cK^\pm_{1,1})^{\pm 1}(\cK^\pm_{1,1})^{\mp 1}=1 \ ,\qquad 
R\cK_1^\pm R^{(0)}\cK_2^\pm = \cK_2^\pm  R^{(0)}\cK_1^\pm R \ ,\qquad 
R\cK_1^+ R^{(0)}\cK_2^- = \cK_2^- R^{(0)}\cK_1^+ R \ .\label{FMint}
\eeqa
This matrix realization is called a presentation of Freidel-Maillet type\footnote{The family of algebras introduced by L. Freidel and J-M. Maillet \cite{FM91} generalizes the concept of Yang-Baxter \cite{FRT87,FRT89} and reflection algebras \cite{Skly88}. They are of the general form \cite[eq. (14)]{FM91}
\beqa
R {\cal K}_1 R' {\cal K}_2 = {\cal K}_2 R'' {\cal K}_1 R'''
\eeqa
where $R,R',R'',R'''$ are non-trivial and subject to certain relations.}.
The Hopf algebra structure is investigated in light of this presentation, and an interpretation of K-operators as intertwiners of certain $U_q(sl_2)$-modules is discussed.  Also, it is found that the simplest example of non-homogeneous Borel subalgebra of $U_q(sl_2)$ considered in \cite{LV17,Vo,LV19} fits in this framework. Some consequences for the spectral parameter dependent Freidel-Maillet equation (\ref{RE}) and corresponding K-operator solutions are studied.  \vspace{1mm} 

The text is organized as follows. Section \ref{sec2} collects the basic material for our purpose. The quantum algebra $U_q(gl_2)$ is recalled,  as well as the Chevalley and equitable presentations of $U_q(sl_2)$. Section \ref{sec3} is devoted to the Freidel-Maillet type presentation of $U_q(sl_2)$.
 It is naturally suggested from a detailed analysis of K-operator solutions built from $U_q(gl_2)$ L-operators introduced at the begining of the same section. The Hopf algebra structure is discussed, and a set of intertwining relations satisfied by the K-operators are obtained. A Freidel-Maillet type presentation for the non-homogeneous Borel subalgebra of $U_q(sl_2)$ follows. 
In Section \ref{sec4}, the K-operator (\ref{Kg}) solution of the spectral parameter dependent Freidel-Maillet equation (\ref{RE}) recently obtained in \cite{Bas20} is  considered in light of (\ref{FMint}). This K-operator characterizes a quotient of the alternating presentation for $U_q(\widehat{sl_2})$ introduced in \cite{Ter18}. It is shown that the class of solutions generated from $\cK^\pm$ follow from specializations to $U_q(sl_2)$ of this K-operator. 
For the simplest irreducible finite dimensional representations in the Chevalley or equitable presentation of $U_q(sl_2)$, explicit examples of spectral parameter dependent K-matrices are obtained. They satisfy a Freidel-Maillet equation of the form:
\begin{align}
R_{12}\cK_{13}R^{(0)}_{12}\cK_{23}=\cK_{23}R^{(0)}_{12}\cK_{13}R_{12}\ ,\label{FMeqint}
\end{align}
and naturally generate quantum integrable models associated with $U_q(sl_2)$. For the Chevalley case, some examples have already been studied in the literature, see e.g. \cite{H94,Ku95,FR02,LIL14}.  \vspace{1mm}

Let us briefly comment on some perspectives and related subjects. Besides potential applications to quantum integrable systems, it is expected that the presentation of Freidel-Maillet type  for $U_q(sl_2)$ (\ref{FMint}) generalizes to any higher rank Lie algebra, thus prodiving a unified framework for the Chevalley and equitable presentation of $U_q(g)$. Also, let us point out that the equitable presentation of $U_q(\widehat{sl_2})$ is known \cite{IT07}.  So, it is natural to ask for the corresponding Freidel-Maillet type presentation of $U_q(\widehat{sl_2})$. This is closely related with the subject of \cite{Bas20,Bas21}. Note that the present work is also motivated by the theory of quantum symmetric pairs (see e.g. \cite{Letzter,MRS,Kolb,BW,LV19}), reflection algebras \cite{Skly88,KS,DKM03,Ko0,KoS,BKo15,RV16,Kolb17,AV}, the q-Onsager algebra \cite{Ter99,B04} and its alternating central extension, see \cite{Ter21} and references therein. As a last comment, for a Hopf algebra recall that the quasi-cocommutativity property is characterized by the universal $R$-matrix and the concept of quasi-triangularity \cite{Dr0}. In view of the results of subsection \ref{sec:intert} it seems desirable to investigate the universal K-matrix of equitable type associated with a universal version of (\ref{FMeqint}), and clarify its interpretation with respect to $U_q(sl_2)$-modules. Some of these problems will be considered elsewhere.    

\vspace{2mm}

{\bf Notation 1.1.}
{\it
 Let ${\mathbb K}$ denote an algebraically closed
field of characteristic $0$. ${\mathbb K}(q)$ denotes the field of rational functions in an indeterminate $q$. The $q$-commutator $\big[X,Y\big]_q=qXY-q^{-1}YX$ is introduced. We denote $[X,Y]=[X,Y]_{q=1}$.  We also denote the $q$-number $[n]_q= (q^n-q^{-n})/(q-q^{-1})$.} \vspace{2mm}


\section{The quantum algebras $U_q(gl_2)$ and $U_q(sl_2)$}\label{sec2}
In this section, we recall the definitions and basic properties of the Chevalley presentations of $U_q(gl_2)$ and $U_q(sl_2)$, as well as the equitable presentation of $U_q(sl_2)$ recently introduced in \cite{ITW}. For each presentation, the comultiplication,  counit and antipode that ensure the Hopf algebra structure are given. Central elements are also displayed. For further discussion, irreducible finite dimensional representations 
\cite{CPb,ITW,Ter09} are recalled.  
\subsection{The Chevalley presentation of $U_q(gl_2)$} We refer the reader to  \cite{Jim86}.
%
\begin{defn} 
\label{def:uqgl}
$U_q(gl_2)$
is the unital associative ${\mathbb K}(q)$-algebra 
with {\it Chevalley}
generators $E,F,K^{\pm 1}_1,K_2^{\pm 1}$ satisfying the following relations:
\begin{eqnarray}
K_iK_i^{-1} &=& K_i^{-1}K_i =  1\ , \quad i=1,2\ ,\label{eqgl1} \\
K_iK_j&=&K_jK_i \ ,\label{eqgl2}\\
K_1EK_1^{-1} &=& qE,\qquad \quad K_1FK_1^{-1}=q^{-1}F\ ,\label{eqgl3}\\
K_2EK_2^{-1} &=& q^{-1}E,\qquad K_2FK_2^{-1}=qF\ ,\label{eqgl4}\\
\big[E,F\big]&=& \frac{K_1K_2^{-1}-K_1^{-1}K_2}{q-q^{-1}}\ .\label{eqgl5}
\end{eqnarray}
\end{defn}
In the text, the multiplication and unit map of $U_q(gl_2)$ will be denoted $\mu: U_q(gl_2) \otimes  U_q(gl_2)  \rightarrow U_q(gl_2) $ and $\eta: {\mathbb K} \rightarrow U_q(gl_2)$, respectively.\vspace{1mm}

The quantum algebra $U_q(gl_2)$  is endowed with the following Hopf algebra structure. As a coalgebra, it is equipped with the comultiplication $\Delta: U_q(gl_2) \rightarrow U_q(gl_2) \otimes U_q(gl_2)$ such that
\begin{eqnarray}
&& \Delta(E) = E \otimes 1+K_1K_2^{-1} \otimes E\ ,\quad \qquad 
\Delta(F) = F \otimes K_1^{-1}K_2+ 1 \otimes F\ ,\qquad  \quad \Delta(K_i) = K_i \otimes K_i\ ,\label{cpgl2}
\end{eqnarray}
and the counit $\varepsilon: U_q(gl_2) \rightarrow {\mathbb K}$ such that
\begin{eqnarray}
\varepsilon(E)=0\ , \qquad \qquad \varepsilon(F)=0\ ,
\qquad \qquad 
\varepsilon(K_i)=1\ .\label{cugl2}
\end{eqnarray}
As a Hopf algebra, in addition it is equipped with the antipode $S: U_q(gl_2)\rightarrow U_q(gl_2)$ such that
\begin{eqnarray}
S(E)= -K_1^{-1}K_2E\ ,
\qquad \qquad 
S(F)= -FK_1K_2^{-1}\ ,
\qquad \qquad 
S(K_i) = K_i^{-1}\ .\label{agl2}
\end{eqnarray}

The center of $U_q(gl_2)$ is generated by two central elements, given by:
\beqa
\Omega_{1,c} &=& K_1K_2 \ ,\label{Centgl21c}\\
\Omega_{2,c} &=& \frac{q^{-1}K_1K_2^{-1}+ q K_1^{-1}K_2}{(q-q^{-1})^2} + EF = \frac{qK_1K_2^{-1}+ q^{-1}K_1^{-1}K_2}{(q-q^{-1})^2} + FE .\label{Centgl22c}
\eeqa

\subsection{The Chevalley and equitable presentations of $U_q(sl_2)$}
Two different presentations of the quantum algebra $U_q(sl_2)$ in terms of generators and relations have been introduced in the literature. 
The Chevalley presentation of $U_q(sl_2)$  is first recalled \cite{KR83,Sk85}.  For the notations, see e.g. \cite[p.~9]{Jantzen}.
\begin{defn} 
\label{def:uq}
$U_q(sl_2)$
is the unital associative ${\mathbb K}(q)$-algebra 
with {\it Chevalley} generators $E,F,K, K^{-1}$ satisfying the following relations:
\begin{eqnarray}
KK^{-1} &=& 
K^{-1}K =  1,
\label{eq1}
\\
KEK^{-1} &=& q^2E,\qquad  KFK^{-1}=q^{-2}F,
\label{eq3}
\\
\big[E,F\big]&=& \frac{K-K^{-1}}{q-q^{-1}}.
\label{eq4}
\end{eqnarray}
\end{defn}

The quantum algebra $U_q(sl_2)$  has the following Hopf algebra structure.  The comultiplication $\Delta: U_q(sl_2) \rightarrow U_q(sl_2) \otimes U_q(sl_2)$ is such that
\begin{eqnarray}
&& \Delta(E) = E \otimes 1+K \otimes E\ ,\quad
\Delta(F) = F \otimes K^{-1}+ 1 \otimes F\ , \quad \Delta(K) = K \otimes K\ , \label{copc}
\end{eqnarray}
The counit $\varepsilon: U_q(sl_2) \rightarrow {\mathbb K}(q)$ is such that
\begin{eqnarray}
\varepsilon(E)=0\ , \qquad \qquad \varepsilon(F)=0\ ,
\qquad \qquad 
\varepsilon(K)=1\  \label{cuc}
\end{eqnarray}
and the antipode $S: U_q(sl_2)\rightarrow U_q(sl_2)$ is such that
\begin{eqnarray}
S(E)= -K^{-1}E\ ,
\qquad \qquad 
S(F)= -FK\ ,
\qquad \qquad 
S(K) = K^{-1}\ . \label{ac}
\end{eqnarray}
Also, note the automorphism: 
\beqa
\theta: \qquad E \leftrightarrow F\ , \qquad K \leftrightarrow K^{-1}\ .\label{auto}
\eeqa

The central element of $U_q(sl_2)$  is the so-called Casimir operator, given by:
\beqa
\Omega_c = \frac{q^{-1}K+ q K^{-1}}{(q-q^{-1})^2} + EF = \frac{qK+ q^{-1}K^{-1}}{(q-q^{-1})^2} + FE .\label{Centc}
\eeqa

Note that in the Chevalley presentation, $U_q(sl_2)$ is the quotient of $U_q(gl_2)$ by the relation $K_1K_2=1$  (see (\ref{Centgl21c})). It is is generated by
\beqa
E \ ,\qquad F \ ,\qquad K=K_1K_2^{-1}\ ,\qquad K^{-1}=K_1^{-1}K_2 \  .\label{isoglsl}
\eeqa
%

In further sections, we will consider examples of irreducible finite dimensional representations of $U_q(sl_2)$. For the Chevalley presentation, see e.g. \cite{CPb}. Choose ${\mathbb K}={\mathbb C}$. Define the vector space  $V^{(s)}$ of dimension $2s+1$ with basis $\{v_0,...,v_{2s}\}$. In this basis, the Chevalley generators act as matrices with non-vanishing entries $(i,j)$, $0 \leq i,j \leq 2s$, such that 
\beqa
(\rho_{V^{(s)}}(K))_{i,i} = q^{2s-2i}\ ,\qquad (\rho_{V^{(s)}}(E))_{i-1,i} = [2s-i+1]_q \ , \qquad (\rho_{V^{(s)}}(F))_{i,i-1} = [i]_q  \ .  \nonumber
\eeqa
For instance, for the spin$-1/2$ representation of $U_q(sl_2)$ one has:
\beqa
\rho_{V^{(1/2)}}(K)=\begin{pmatrix} q & 0   \\
   0 & q^{-1} \end{pmatrix}  \ ,\qquad \rho_{V^{(1/2)}}(E)=\begin{pmatrix} 0 & 1   \\
   0 & 0 \end{pmatrix} \ ,\qquad \rho_{V^{(1/2)}}(F)=\begin{pmatrix} 0 & 0   \\
   1 & 0 \end{pmatrix}    \ .\label{repec}
\eeqa

\vspace{2mm}

More recently, the {\it equitable} presentation of the quantum algebra $U_q(sl_2)$ has been introduced.  
\begin{thm}
\cite[Theorem 2.1]{ITW}
\label{thm1}
The algebra $U_q(sl_2)$ is isomorphic to the unital associative ${\mathbb K}(q)$-algebra 
with {\it equitable}  generators  $X$, $X^{-1}$, $Y$, $Z$ satisfying the following relations:
\begin{eqnarray}
XX^{-1} = 
X^{-1}X &=&  1\ ,
\label{eq:e1}
\\
\frac{qXY-q^{-1}YX}{q-q^{-1}}&=&1\ ,\qquad 
\frac{qYZ-q^{-1}ZY}{q-q^{-1}}=1\ ,\qquad
\frac{qZX-q^{-1}XZ}{q-q^{-1}}=1\ . \label{eq:e4}
\end{eqnarray}
\end{thm}
The isomorphism with the Chevalley presentation of Definition \ref{def:uq} is given by:
\beqa
\label{iso1}
\phi: \quad X^{{\pm}1} &\rightarrow & K^{{\pm}1}\ ,\qquad 
Y \rightarrow  K^{-1}+(q-q^{-1})F\ , \qquad
Z \rightarrow  K^{-1}-q(q-q^{-1})K^{-1}E\ .
\eeqa
Using the automorphism $\theta$, one gets another isomorphism:
\beqa
\label{iso2}
X^{{\pm}1} &\rightarrow & K^{\mp 1}\ ,\qquad 
Y \rightarrow  K+(q-q^{-1})E\ , \qquad
Z \rightarrow  K-q(q-q^{-1})KF\ .
\eeqa
Note that the above isomorphisms can be generalized combining automorphisms of the Chevalley presentation of $U_q(sl_2)$. For instance \cite[Lemma 2.3]{ITW}: 
\beqa
E \rightarrow aEK^\ell\ , \qquad F \rightarrow a^{-1}K^{-\ell}F\ , \qquad K \rightarrow K\ , \vspace{1mm}
\eeqa
for any integer $\ell$ and $a\in {\mathbb K}$.\vspace{1mm}

The Hopf algebra structure is ensured by the comultiplication $\Delta$   such that 
\begin{eqnarray}
&&\Delta(X) = X \otimes X\ ,\qquad \Delta(Y) = (Y-1) \otimes X^{-1}+
1 \otimes Y\ ,\qquad
\Delta(Z) = 
(Z-1) \otimes X^{-1}+
1 \otimes Z \ .\label{cope}
\end{eqnarray}
The counit $\varepsilon$ is such that:
\begin{eqnarray}
\varepsilon(X)=1\ , \qquad \qquad \varepsilon(Y)=1\ ,
\qquad \qquad 
\varepsilon(Z)=1 \label{cue}
\end{eqnarray}
and the antipode $S$ is such that:
\begin{eqnarray}
S(X)= X^{-1}\ ,
\qquad \qquad 
S(Y)= 1+X-YX\ ,
\qquad \qquad 
S(Z) = 1+X-ZX\ .\label{ae}
\end{eqnarray}

Note that the subalgebra $\cal B$ generated by $\{Y,Z\}$ gives an example of {\it non-homogeneous} Borel subalgebra, see \cite[Example 3.6]{LV19}. There is no analog of $\cal B$ in $sl_2$ (specialization $q\rightarrow 1$).
As can be seen from (\ref{cope}), $\cal B$ is a right coideal subalgebra of $U_q(sl_2)$:
 $\Delta(\cal B) \subset \cal B \otimes U_q(sl_2)$. 
On this subject, see \cite{HK,LV17,Vo,LV19}. \vspace{1mm}

In the equitable presentation of $U_q(sl_2)$, the central element is given by:
\beqa
\Omega_e = qX + q^{-1}Y + qZ - qXYZ .\label{Cente}
\eeqa
Several equivalent expressions can be derived using (\ref{eq:e4}). 
Note that using the inverse of (\ref{iso1}), one finds $\Omega_c \rightarrow \Omega_e/(q-q^{-1})^2$.
\vspace{1mm}

For the equitable presentation, irreducible finite dimensional representations have been studied in details in \cite{ITW,Ter09}. Adapting the notations of \cite{Ter09} (see Theorem 10.12), choose ${\mathbb K}={\mathbb C}$ and consider the vector space  $V^{(s)}$.
Among the twelve known bases,  choose the one  denoted  $\{u_0,...,u_{2s}\}$ of type `$[y]_{col}$' in \cite[Theorem 10.12]{Ter09}. In this basis, the equitable generators act as matrices with non-vanishing entries $(i,j)$, $0 \leq i,j \leq 2s$ such that 
\beqa
(\rho_{V^{(s)}}(X))_{i,i} &=& q^{2s-2i}\ ,\nonumber\\
(\rho_{V^{(s)}}(Y))_{i,i} &=& q^{2i-2s}\ ,\qquad (\rho_{V^{(s)}}(Y))_{i,i-1} = q^{2s} - q^{2i-2-2s} \ ,\nonumber\\ 
(\rho_{V^{(s)}}(Z))_{i,i} &=& q^{2i-2s}\ ,\qquad (\rho_{V^{(s)}}(Z))_{i-1,i} = q^{-2s} - q^{2i-2s}  \ .  \nonumber
\eeqa
For instance, for the spin$-1/2$ representation of $U_q(sl_2)$ one has:
\beqa
&& \rho_{V^{(1/2)}}(X)=\begin{pmatrix} q & 0   \\
   0 & q^{-1} \end{pmatrix}  \ ,\qquad \rho_{V^{(1/2)}}(Y)=\begin{pmatrix} q^{-1} & 0 \\
   q-q^{-1}   & q \end{pmatrix} \ ,\qquad \rho_{V^{(1/2)}}(Z)=\begin{pmatrix} q^{-1} & q^{-1}-q   \\
   0 & q \end{pmatrix}    \ .\label{repeq}
\eeqa

\section{Freidel-Maillet type presentations for $U_q(sl_2)$}\label{sec3}
In this section, we construct quantum K-operators that satisfy a system of Freidel-Maillet equations associated with  $U_q(gl_2)$. Then, a presentation of Freidel-Maillet  type for $U_q(sl_2)$ in terms of K-operators of Chevalley or equitable type is proposed. The comultiplication,  counit and antipode of $U_q(sl_2)$ are identified, and central elements are given. In particular, the non-homogeneous Borel subalgebra of $U_q(sl_2)$ fits in this framework. A set of intertwining relations satisfied by the K-operators is obtained. Some basic examples of constant K-matrices of Chevalley and equitable type are also displayed.\vspace{1mm}

As a preliminary, we first recall the well-known FRT presentation  of $U_q(gl_2)$ that will be used to construct a one-parameter family of K-operators whose entries are in a subalgebra of $U_q(gl_2)$.    
\subsection{The FRT presentation of $U_q(gl_2)$}\label{sec31}
We refer the reader to \cite{FRT87,FRT89,DF93} for details. Define the quantum R-matrix 
\begin{align}
R =\left(
\begin{array}{cccc} 
 q    & 0 & 0 & 0 \\
0  &  1 & q-q^{-1} & 0 \\
0  &  0 & 1 &  0 \\
0 & 0 & 0 & q
\end{array} \right) \ \label{Rgl}
\end{align}
with deformation parameter $q$.
It is known that $R$ satisfies the quantum Yang-Baxter equation (\ref{YB0}) in the space ${\cal V}_1\otimes {\cal V}_2\otimes {\cal V}_3$. 
%
%
%
%
In addition, $R$ satisfies $R-PR^{-1}P=(q-q^{-1})P$, where $P$ is the permutation operator $P(v_1 \otimes v_2) = v_2 \otimes v_1$, $v_1,v_2\in {\mathbb C}^2$. Note that $R^{-1}=R_{q\leftrightarrow q^{-1}}$ and $R^{t_1t_2}=PRP=R_{21}$. Recall the notation $\cL_1=\cL \otimes {\mathbb I}$, $\cL_2= {\mathbb I} \otimes \cL$.
\begin{thm} \label{thm:isoUglFRT}
\cite{FRT89,DF93}
 $U_q(gl_2)$ admits a presentation of FRT type. Define the quantum L-operators with entries $\cL^\pm_{i,j}$, $i,j=1,2$: 
\beqa
 \cL^+=\begin{pmatrix} K_1 & (q-q^{-1})K_1F   \\
   0 & K_2
      \end{pmatrix} \ ,\qquad  \cL^-=\begin{pmatrix} K_1^{-1} & 0   \\
   -(q-q^{-1})EK_1^{-1} & K_2^{-1} \label{Lpm}
      \end{pmatrix} \ .
\eeqa
%
The defining relations are given in a matrix form as follows:
\beqa
\cL^+_{i,i}\cL^-_{i,i}&=&\cL^-_{i,i}\cL^+_{i,i}=1 \ ,\label{frt1}\\
R\cL_1^\pm \cL_2^\pm &=& \cL_2^\pm \cL_1^\pm R \ ,\label{frt2}\\
R\cL_1^+ \cL_2^- &=& \cL_2^- \cL_1^+ R \ .\label{frt3}
\eeqa
\end{thm}
From (\ref{frt2}), (\ref{frt3}), one shows: 
\beqa
R_{21}^{-1}\cL_1^\pm \cL_2^\pm= \cL_2^\pm \cL_1^\pm R_{21}^{-1}\ ,\qquad
R_{21}^{-1}\cL_1^- \cL_2^+= \cL_2^+ \cL_1^- R_{21}^{-1}\ .\label{frt4}
\eeqa
\begin{rem} The inverse L-operators are given by:
\beqa
 (\cL^+)^{-1}=\begin{pmatrix} K_1^{-1} & -(q-q^{-1})FK_2^{-1}   \\
   0 & K_2^{-1}
      \end{pmatrix} \ ,\qquad  (\cL^-)^{-1}=\begin{pmatrix} K_1 & 0   \\
   (q-q^{-1})K_2E & K_2
      \end{pmatrix} \ .
\eeqa
\end{rem}
The relations satisfied by the inverse L-operators are easily derived from (\ref{frt2})-(\ref{frt3}), see \cite[Proposition 2.1]{DF93}. For instance,
\beqa
R_{21} (\cL^\pm)^{-1}_1(\cL^\pm)^{-1}_2 &=& (\cL^\pm)^{-1}_2 (\cL^\pm)^{-1}_1 R_{21} \ ,\label{invRTT1}\\
R_{21} (\cL^-)^{-1}_1(\cL^+)^{-1}_2 &=& (\cL^+)^{-1}_2 (\cL^-)^{-1}_1 R_{21} \ .\label{invRTT2}
\eeqa
%
%
%

In the FRT presentation of Theorem \ref{thm:isoUglFRT}, the Hopf algebra structure of $U_q(gl_2)$ is characterized as follows. The corresponding expressions for the comultiplication (\ref{cpgl2}), counit (\ref{cugl2}) and antipode (\ref{agl2})  are defined by\footnote{$((T)_{[\textsf 1]}(T')_{[\textsf 2]})_{ij} =\sum_{k=1}^2 (T)_{ik} \otimes  (T')_{kj}$ \cite{Skly88}.}
\beqa
\Delta(\cL^\pm) = (\cL^\pm)_{[1]}  (\cL^\pm)_{[2]}  \ ,\qquad \varepsilon(\cL^\pm)= {\mathbb I} \ ,\qquad S(\cL^\pm)= (\cL^\pm)^{-1}  \ ,\label{ruleFRT}
\eeqa
where ${\mathbb I}$ denotes the $2 \times 2$ identity matrix.
Note that due to the comultiplication rule, the last expressions easily follow from the defining property of the antipode  $\mu \circ (id \otimes S) \circ \Delta = \mu \circ (S \otimes id) \circ \Delta = \eta \circ \varepsilon$, which gives $\cL^\pm S(\cL^\pm)= S(\cL^\pm)\cL^\pm = {\mathbb I}$. \vspace{1mm}

%
%
%

Central elements of  $U_q(gl_2)$ are derived from the following quantum determinants. Define $U=PR-q({\mathbb I} \otimes {\mathbb I})  $ which satisfies the Hecke braid relation $U_{12}U_{23}U_{12}=U_{12}$ and $U^2 =-(q+q^{-1})U$. As usual, below `$\rm tr_{12}$' stands for the trace over $\cal V_1 \otimes \cal V_2$.
One has:
\beqa
{\rm det}_q^{(1)} &=& {\rm tr}_{12}(U_{12} \cL_1^+\cL_2^+) =  - (q+q^{-1})\Omega_{1,c} \ ,\label{qdet1}\\
{\rm det}_q^{(2)} &=& {\rm tr}_{12}(U_{12} \cL_1^+\cL_2^-) = -(q-q^{-1})^2  \Omega_{2,c} \ ,\label{qdet2}
\eeqa
with (\ref{Centgl21c}), (\ref{Centgl22c}). \vspace{1mm}

For further discussion, let us introduce $U'_q(sl_2)$ the extension of the Chevalley presentation of $U_q(sl_2)$ by the elements $K^{\pm 1/2}$ \cite{FMu}. The defining relations of $U'_q(sl_2)$ are given by (\ref{eqgl1})-(\ref{eqgl5}) with the substitution 
\beqa
K_1 \rightarrow K^{1/2}\ , \qquad K_2 \rightarrow K^{-1/2} \ .\label{Kdemi}
\eeqa
A FRT presentation for $U'_q(sl_2)$ is obtained as a corollary of Theorem \ref{thm:isoUglFRT}: it is defined as the quotient of (\ref{frt1})-(\ref{frt3}) by $\rm det(\cL^\pm)=1$. See e.g. \cite{DF93,FMu} for details. The corresponding quantum L-operators are given by (\ref{Lpm}) with (\ref{Kdemi}).

\subsection{K-operators and $U_q(gl_2)$}
Explicit examples of `dressed' solutions of Freidel-Maillet type equations can be constructed using the FRT presentation of $U_q(gl_2)$ given in Theorem \ref{thm:isoUglFRT}. The analysis below is done by analogy with \cite[Proposition 2]{Skly88}, so we only sketch the main steps of the proofs. Consider the R-matrix (\ref{Rgl}) and introduce the diagonal matrix: 
\beqa
 R^{(0)}=diag(1,q^{-1},q^{-1},1)\ .\label{R0}
\eeqa
Consider the Freidel-Maillet equations:
\beqa
R\ (\cK\otimes {\mathbb I})\ R^{(0)}\ ({\mathbb I} \otimes \cK)\
&=& \ ({\mathbb I} \otimes \cK)\  R^{(0)}\ (\cK\otimes {\mathbb I})\ R\ ,\label{FMKKgl}\\
R\ (\cK\otimes {\mathbb I})\ R^{(0)}\ ({\mathbb I} \otimes \cK')\
&=& \ ({\mathbb I} \otimes \cK')\  R^{(0)}\ (\cK\otimes {\mathbb I})\ R\ .\label{FMKKpgl}
\eeqa
For the following analysis, we need the `reduced' quantum Lax operators:
\beqa
\cL^{0,+}=diag(K_1,K_2)  \ ,\qquad \bar \cL^{0,-}=diag(K_2^{-1},K_1^{-1})\ \label{redL}
\eeqa
and the elementary K-matrices 
\beqa
\cK^{0,\alpha} =  \begin{pmatrix} 1 & 0  \\
  \alpha  & 1 \end{pmatrix}\ , \label{K0alpha}
\eeqa
where the scalar $\alpha \in {\mathbb K}$ is introduced. Combining (\ref{Lpm}), (\ref{redL}) and (\ref{K0alpha}), `dressed' K-operators satisfying the system (\ref{FMKKgl}),(\ref{FMKKpgl}), are easily obtained. 
\begin{lem}\label{lemRKR0K} 

For any $\alpha \in {\mathbb K}$, the K-operators
\beqa
\cK^{+,\alpha}= \bar \cL^{0,-}\cK^{0,\alpha}\cL^{+} \qquad \mbox{and} \qquad \cK^{-,\alpha} = \cL^{0,+}\cK^{0,\alpha}\cL^{-}\  \label{solKdef}
\eeqa
satisfy the Freidel-Maillet equations:\vspace{1mm}

(i) (\ref{FMKKgl}) for $\cK \rightarrow \cK^{+,\alpha}$ or $\cK \rightarrow \cK^{-,\alpha}$;\\

\vspace{-2mm}

(ii) (\ref{FMKKpgl}) for $\cK \rightarrow \cK^{+,\alpha}$ and $\cK' \rightarrow \cK^{-,\alpha}$.


%
\end{lem}
\begin{proof} Consider (i). Assume there exists a matrix $\cK^0$, two quantum Lax operators $\cL , \cL^{0}$ and a R-matrix $R'$ such that the following relations hold:
\beqa  R \   \cK^0_1 \ R'\  \cK^0_2\
&=& \   \cK^0_2  \ R' \  \cK^0_1 \  R\  ,\label{RKinit} \\
 R    \cL_1  \cL_2   &=& \cL_2   \cL_1 R \ ,\label{RLpLp}\\
 R   (\cL^{0})_1(\cL^{0})_2 &=& (\cL^{0})_2(\cL^{0})_1  R \ ,\label{RL0L0}\\
(\cL^{0})_1R'  \cL_2  &=&  \cL_2   R' (\cL^{0})_1\ ,\label{L0R0L}\\
\cL_1 R'(\cL^{0})_2  &=& (\cL^{0})_2R' \cL_1  \  \label{LR0L0}
\eeqa
and
\beqa
\big[\cK^0_i, \cL_j \big]=0 \ , \qquad \big[\cK^0_i, (\cL^0)_j \big]=0\ , \qquad i\neq j\ .\label{comKL}
\eeqa
Recall (\ref{Lpm}), (\ref{redL}), (\ref{K0alpha}). For the choices
\beqa
R' = R^{(0)}\ ,\quad  \cL \rightarrow \cL^+\ ,\quad \cL^{0} \rightarrow \bar \cL^{0,-} \quad \mbox{and} \quad \cK^{0} \rightarrow \cK^{0,\alpha}  \ ,
\eeqa
one finds that all above relations are satisfied. Adapting \cite[Proposition 2]{Skly88},  using those relations one finds  
$\bar \cL^{0,-}\cK^{0,\alpha}\cL^{+}$ satisfies  (\ref{FMKKgl}). For the choices
\beqa
R' = R^{(0)}\ ,\quad  \cL \rightarrow \cL^-\ ,\quad \cL^{0} \rightarrow \cL^{0,+} \quad \mbox{and} \quad \cK^{0} \rightarrow \cK^{0,\alpha}  \ ,
\eeqa
one finds   $\cL^{0,+}\cK^{0,\alpha}\cL^{-}$ also satisfies (\ref{FMKKgl}). This completes the proof of (i).

To show (ii), it is sufficient to check that the basic equations (\ref{RKinit})-(\ref{LR0L0}) and (\ref{comKL}) hold for
the substitutions $R' = R^{(0)}$ and
\beqa
(\cL^{0})_1 \rightarrow (\bar \cL^{0,-})_1 \quad \mbox{and} \quad (\cL^{0})_2 \rightarrow (\cL^{0,+})_2\ .
\eeqa
%
\end{proof}

Explicitly, the K-operators (\ref{solKdef}) read
\beqa
&& \quad \cK^{+,\alpha} =  \begin{pmatrix} K_1K_2^{-1} & (q-q^{-1})K_1K_2^{-1}F  \\
   \alpha & K_2K_1^{-1} + \alpha(q-q^{-1})F \end{pmatrix}\quad  \mbox{and} \quad  
	\cK^{-,\alpha} = \begin{pmatrix} 1 &  0 \\
  \alpha K_2K_1^{-1}-q(q-q^{-1})K_2K_1^{-1}E & 1 \end{pmatrix}\ .
\  \label{solKexp}
\eeqa
Note that although the K-operators are built from $U_q(gl_2)$, it is seen from the explicit form (\ref{solKexp}) that their entries belong to the subalgebra generated by (\ref{isoglsl}). 
\vspace{1mm}

To prepare the analysis in the next subsection,  let us first observe that more general `dressed' solutions  of (\ref{FMKKgl}) and (\ref{FMKKpgl}) can be easily derived by induction. The proof of the following Lemma\footnote{The notation $((T)_{[\textsf 2]}(T')_{[\textsf 1]}(T'')_{[\textsf 2]})_{ij} =\sum_{k,\ell=1}^2 (T')_{k\ell} \otimes  (T)_{ik}(T'')_{\ell j}$ is used.} can be done similarly to Lemma \ref{lemRKR0K}, thus we skip it. 
\begin{lem} \label{lemcopgl}
The assertions (i), (ii), in Lemma \ref{lemRKR0K} hold for the substitution
%
\beqa
&&\cK^{+,\alpha} \rightarrow 
\bar\cK^{+,\alpha} = 
(\bar \cL^{0,-})_{[2]}(\cK^{+,\alpha})_{[1]}(\cL^{+})_{[2]}   \ ,\qquad \cK^{-,\alpha} \rightarrow 
 \bar\cK^{-,\alpha} = (\cL^{0,+})_{[2]}(\cK^{-,\alpha})_{[1]}(\cL^{-})_{[2]}\  .\label{copgl}
\eeqa
\end{lem} 
Explicitly, the entries of the K-operators  
 read as follows:
\beqa
(\bar\cK^{+,\alpha})_{1,1}&=& K_1K_2^{-1} \otimes K_1K_2^{-1} \ ,\nonumber \\
(\bar\cK^{+,\alpha})_{2,2}&=&  \left( K_2K_1^{-1} +\alpha (q-q^{-1})F -\alpha \right) \otimes K_2K_1^{-1} + 1 \otimes \left( \alpha K_2K_1^{-1} +\alpha (q-q^{-1})F \right)\ ,\nonumber\\
(\bar\cK^{+,\alpha})_{1,2}&=& K_1K_2^{-1} \otimes (q-q^{-1})K_1K_2^{-1}F + (q-q^{-1})K_1K_2^{-1}F \otimes 1 \ ,\nonumber\\
(\bar\cK^{+,\alpha})_{2,1}&=& \alpha (1\otimes 1)\ \nonumber
\eeqa
and
\beqa
(\bar\cK^{-,\alpha})_{1,1}&=& ((\bar\cK^{-,\alpha}))_{2,2} = 1 \otimes 1 \ ,\nonumber\\
(\bar\cK^{-,\alpha})_{1,2}&=& 0 \ ,\nonumber\\
(\bar\cK^{-,\alpha})_{2,1}&=& \left( \alpha (K_2K_1^{-1} -1) -q(q-q^{-1})K_2K_1^{-1}E \right) \otimes K_2K_1^{-1} + 1 \otimes \left( \alpha K_2K_1^{-1} -q(q-q^{-1})K_2K_1^{-1}E \right) \ .\nonumber\label{coexp4}
\eeqa

The action of Hopf algebra homomorphisms $\Delta,\varepsilon$ and antiautomorphism $S$ of $U_q(gl_2)$ on the K-operators  is finally considered. 
Acting with (\ref{cpgl2}) (resp. (\ref{cugl2})) on the K-operators in  (\ref{solKexp}), one compares the resulting expressions with the entries $(\tilde{\cK}^{\pm,\alpha})_{i,j}$ (resp. $(\cK^{0,\alpha})_{i,j}$ with (\ref{K0alpha})). One concludes:
\beqa
\Delta(\cK^{+,\alpha}) &=&  (\bar \cL^{0,-})_{[2]}(\cK^{+,\alpha})_{[1]}(\cL^{+})_{[2]}   \ ,\qquad \Delta(\cK^{-,\alpha})  
 = (\cL^{0,+})_{[2]}(\cK^{-,\alpha})_{[1]}(\cL^{-})_{[2]}\ ,\label{coKpmgl}\\
\varepsilon(\cK^{\pm,\alpha}) &=& \cK^{0,\alpha} \ \label{cuKpmgl}
\eeqa
Furthermore, using (\ref{agl2}) one checks the compatibility property:
\beqa
\mu \circ(id \otimes S) \circ \Delta(\cK^{\pm,\alpha}) = \mu \circ(S \otimes id) \circ \Delta(\cK^{\pm,\alpha}) = \eta \circ \varepsilon(\cK^{\pm,\alpha}) \ . \label{aKpmgl}
\eeqa
%
 
\subsection{Freidel-Maillet type presentations of $U_q(sl_2)$}
Consider the restriction in $U_q(sl_2)$ of the one-parameter family of K-operators (\ref{solKexp}), using (\ref{isoglsl}). They are denoted 
$(\cK^{\pm,\alpha})|_{U_q(sl_2)}$.
For the specialization $\alpha=0$, 
the entries of $(\cK^{\pm,0})|_{U_q(sl_2)}$ are expressed in terms of the Chevalley generators $E,F,K^{\pm1}$. For $\alpha=1$, applying $\phi^{-1}$ with (\ref{iso1}) the entries are expressed in terms of the equitable generators $X^{\pm 1},Y,Z$. Accordingly, introduce the notation:
\beqa
\cK_c^\pm \equiv (\cK^{\pm,0})|_{U_q(sl_2)} \qquad \mbox{and} \qquad   \cK_e^\pm \equiv \phi^{-1}((\cK^{\pm,1})|_{U_q(sl_2)})\ .\label{notK}
\eeqa
By Lemma \ref{lemRKR0K}, they satisfy the Freidel-Maillet equations (\ref{FMKKgl}), (\ref{FMKKpgl}). Furthermore:
\begin{thm} \label{thm:isoUslFM}

 $U_q(sl_2)$ admits a presentation of Freidel-Maillet type. The defining relations are given in a matrix form by:
\beqa 
\cK^\pm_{1,1}(\cK^\pm_{1,1})^{-1}&=& (\cK^\pm_{1,1})^{-1}\cK^\pm_{1,1} =1 \ ,\label{FM1}\\
R\ (\cK^\pm\otimes {\mathbb I})\ R^{(0)}\ ({\mathbb I} \otimes \cK^\pm)\
&=& \ ({\mathbb I} \otimes \cK^\pm)\  R^{(0)}\ (\cK^\pm\otimes {\mathbb I})\ R\ ,\label{FMpp}\\
R\ (\cK^+\otimes {\mathbb I})\ R^{(0)}\ ({\mathbb I} \otimes \cK^-)\
&=& \ ({\mathbb I} \otimes \cK^-)\ R^{(0)}\ (\cK^+\otimes {\mathbb I})\ R\ ,
\label{FMpm} 
\eeqa
where $\cK^\pm \equiv \cK^\pm_c$ or $\cK^\pm \equiv \cK^\pm_e$ is a square matrix with entries $\cK^\pm_{i,j}$, $i,j=1,2$, such that
\beqa
 \cK^+_c=\begin{pmatrix} K & (q-q^{-1})KF   \\
   0 & K^{-1}
      \end{pmatrix} \ ,\qquad  \cK^-_c=\begin{pmatrix} 1 & 0   \\
   -q(q-q^{-1})K^{-1}E &  1\label{Kpmc}
      \end{pmatrix} \ 
\eeqa
and
\beqa
 \cK^+_e=\begin{pmatrix} X & XY-1   \\
   1 & Y
      \end{pmatrix} \ ,\qquad  \cK^-_e=\begin{pmatrix} 1 & 0   \\
   Z &  1\label{Kpme}
      \end{pmatrix} \ .
\eeqa
We call $\cK^\pm_c$ and $\cK^\pm_e$ the Chevalley and equitable K-operators of $U_q(sl_2)$, respectively. 
\end{thm}
\begin{proof} Eq. (\ref{FM1}) is equivalent to (\ref{eq1}) or (\ref{eq:e1}). Insert (\ref{Kpmc}) or (\ref{Kpme}) into (\ref{FM1})-(\ref{FMpm}). It is checked that the resulting system of equations is equivalent to the remaining relations given in Definition \ref{def:uq} or Theorem \ref{thm1}.
\end{proof}
Note that (\ref{FMpp}), (\ref{FMpm}) imply
\beqa
R_{21}^{-1}\ (\cK^\pm \otimes {\mathbb I})\ R^{(0)}\ ({\mathbb I} \otimes \cK^\pm)\
= \ ({\mathbb I} \otimes \cK^+)\ R^{(0)}\ (\cK^-\otimes {\mathbb I})\ R_{21}^{-1}\ ,\label{RmKKpp}\\
R_{21}^{-1}\ (\cK^-\otimes {\mathbb I})\ R^{(0)}\ ({\mathbb I} \otimes \cK^+)\
= \ ({\mathbb I} \otimes \cK^+)\ R^{(0)}\ (\cK^-\otimes {\mathbb I})\ R_{21}^{-1}\ .\label{RmKmKp}  
\eeqa
\begin{rem} The inverse K-operators are given by:
\beqa
 (\cK_c^+)^{-1}=\begin{pmatrix} K^{-1} & -(q-q^{-1})FK   \\
   0 & K
      \end{pmatrix} \ ,\qquad  (\cK_c^-)^{-1}=\begin{pmatrix} 1 & 0   \\
   q^{-1}(q-q^{-1})E K^{-1}& 1
      \end{pmatrix} \ \label{Kopc}
\eeqa
and
\beqa
 (\cK_e^+)^{-1}=\begin{pmatrix} Y & 1 -YX  \\
   -1 & X
      \end{pmatrix} \ ,\qquad  (\cK_e^-)^{-1}=\begin{pmatrix} 1 & 0   \\
   -Z & 1
      \end{pmatrix} \ .\label{Kope}
\eeqa
\end{rem}

By analogy with (\ref{invRTT1}), (\ref{invRTT2}),
other examples of Freidel-Maillet equations can be derived from (\ref{FMpp})-(\ref{FMpm}).
\begin{prop} 
The following relations hold:
\beqa
R_{21} (\cK^\pm)^{-1}_1 (R^{(0)})^{-1}(\cK^\pm)^{-1}_2 &=&  (\cK^\pm)^{-1}_2 (R^{(0)})^{-1} (\cK^\pm)^{-1}_1 R_{21} \ ,\\
R_{21} (\cK^-)^{-1}_1 (R^{(0)})^{-1}(\cK^+)^{-1}_2 &=&  (\cK^+)^{-1}_2 (R^{(0)})^{-1} (\cK^-)^{-1}_1 R_{21} \ .
\eeqa
%
%
%
\end{prop}
Note that the Chevalley K-operators have a triangular structure, contrary to one of the two equitable K-operators. \vspace{1mm}

The Freidel-Maillet type presentation of Theorem \ref{thm:isoUslFM} is not unique. For instance, the relation (\ref{eq:e1}) is not necessary to show (\ref{FMpp}), (\ref{FMpm}), given the K-operators (\ref{Kpme}): only the relations (\ref{eq:e4}) are needed. Because these latter relations remain unchanged by the rotation $r: \ X \rightarrow Y$, $Y \rightarrow Z$,  $Z \rightarrow X$, other expressions for K-operators obtained from (\ref{Kpme}) can be considered as well. For instance, denote $r(\cK_e^{+})\equiv \cK_{\cal B}$ and $r(\cK_e^{-})\equiv \cK_{X}$. It follows:
\begin{rem} A Freidel-Maillet type presentation of  $U_q(sl_2)$ is given by:
\beqa 
(\cK_a)_{2,1}(\cK_a)_{2,1}^{-1}&=& (\cK_a)_{2,1}^{-1}(\cK_a)_{2,1} =1 \ ,\label{FMBorel1}\\
R\ (\cK_a\otimes {\mathbb I})\ R^{(0)}\ ({\mathbb I} \otimes \cK_a)\
&=& \ ({\mathbb I} \otimes \cK_a)\  R^{(0)}\ (\cK_a\otimes {\mathbb I})\ R\ \quad \mbox{for}\qquad  a=X,{\cal B},\label{FMBorelBBXX}\\
R\ (\cK_{\cal B}\otimes {\mathbb I})\ R^{(0)}\ ({\mathbb I} \otimes \cK_X)\
&=& \ ({\mathbb I} \otimes \cK_X)\ R^{(0)}\ (\cK_{\cal B}\otimes {\mathbb I})\ R\ ,
\label{FMBorelBX} 
\eeqa
where $\cK_{\cal B}, \cK_X$ are square matrices such that
\beqa
 \cK_{\cal B}=\begin{pmatrix} Y & YZ-1   \\
   1 & Z
      \end{pmatrix} \ ,\qquad  \cK_X=\begin{pmatrix} 1 & 0   \\
   X &  1\label{Kborel}
      \end{pmatrix} \ .
\eeqa
\end{rem}

As mentionned in the introduction, the subalgebra $\cal B$ generated by $\{Y,Z\}$ in the equitable presentation of $U_q(sl_2)$ provides the simplest example of non-homogeneous Borel subalgebras \cite{LV17,Vo,LV19}. This is seen from a direct comparison between (\ref{iso1}) and \cite[Example 3.6]{LV19}. According to previous remark, $\cal B$ is isomorphic to the Freidel-Maillet algebra:
\beqa
R\ (\cK_{\cal B}\otimes {\mathbb I})\ R^{(0)}\ ({\mathbb I} \otimes \cK_{\cal B})\
&=& \ ({\mathbb I} \otimes \cK_{\cal B})\  R^{(0)}\ (\cK_{\cal B}\otimes {\mathbb I})\ R\ .\label{FMKKBorel}
\eeqa

Another Freidel-Maillet type presentation is obtained by modifying the R-matrix $R^{(0)}$: 
\begin{rem} A Freidel-Maillet type presentation of  $U_q(sl_2)$ is obtained from (\ref{FM1})-(\ref{FMpm}) with the substitution $R^{(0)} \rightarrow (R^{(0)})^{-1}$, $\cK^\pm_{1,1}\rightarrow \cK^\pm_{2,2}$ and the Chevalley and equitable K-operators given, respectively, by:
\beqa
 \cK_c^+ \rightarrow \begin{pmatrix} 1 & (q-q^{-1})F  \\
   0 & 1
      \end{pmatrix} \ ,\qquad  \cK_c^- \rightarrow \begin{pmatrix} K^{-1} & 0   \\
  -q(q-q^{-1})KE & K
      \end{pmatrix} \ \label{Kopcbis}
\eeqa
and
\beqa
 \cK_e^+ \rightarrow \begin{pmatrix} 1 & Y  \\
   0 & 1
      \end{pmatrix} \ ,\qquad  \cK_e^- \rightarrow \begin{pmatrix} Z & 1   \\
   XZ-1 & X
      \end{pmatrix} \ .\label{Kopebis}
\eeqa
\end{rem}

For the FRT presentation of $U'_q(sl_2)$ introduced at the end of Subsection \ref{sec31}, the Hopf algebra structure is characterized by a comultiplication, counit and antipode that follow from the restriction of (\ref{ruleFRT}) to $U'_q(sl_2)$. A characterization of the Hopf algebra structure of $U_q(sl_2)$ is obtained by combining the Freidel-Maillet type presentation of Theorem \ref{thm:isoUslFM} and the FRT presentation of $U'_q(sl_2)$ extended by the `reduced' L-operators (\ref{redL}). The comultiplication (\ref{copc}) or (\ref{cope})  and the counit (\ref{cuc}), (\ref{cue}) are obtained as follows. Denote by $M|_{U_q(sl_2)}$ the restriction of the operator $M$ to $U_q(sl_2)$.
%


\begin{prop} 
 
The comultiplication $\Delta: U_q(sl_2) \rightarrow U_q(sl_2) \otimes U_q(sl_2)$ is such that
\beqa
&& \Delta(\cK_c^+)=\left((\bar \cL^{0,-})_{[2]}(\cK_c^{+})_{[1]}(\cL^{+})_{[2]}\right)|_{U_q(sl_2)_{[2]}}   \ ,\qquad \Delta(\cK_c^{-}) =  \left((\cL^{0,+})_{[2]}(\cK_c^{-})_{[1]}(\cL^{-})_{[2]}\right)|_{U_q(sl_2)_{[2]}} \ , \label{DeltaKcpm}
\eeqa
\beqa
\Delta(\cK_e^+)&=& \left(id \otimes \phi^{-1}) \circ (\bar \cL^{0,-})_{[2]}(\cK_e^{+})_{[1]}(\cL^{+})_{[2]}\right)|_{U_q(sl_2)_{[2]}}   \ , \label{DeltaKepm}\\ 
\Delta(\cK_e^{-}) &=& (id \otimes \phi^{-1}) \circ \left((\cL^{0,+})_{[2]}(\cK_e^{-})_{[1]}(\cL^{-})_{[2]}\right)|_{U_q(sl_2)_{[2]}}\ .\nonumber 
\eeqa
%

The counit $\varepsilon: U_q(sl_2) \rightarrow {\mathbb K}$ is such that
\beqa
\varepsilon(\cK_c^\pm)=\cK^{0,0} \qquad \mbox{and}  \qquad \varepsilon(\cK_e^\pm)=\cK^{0,1}\  
\label{cuKcepm}
\eeqa
with (\ref{K0alpha}).
\end{prop}
\begin{proof}   Firstly, we show that (\ref{DeltaKcpm}), (\ref{DeltaKepm}) are equivalent to the 
comultiplication rules (\ref{copc}), (\ref{cope}). For the specialization $\alpha=0$, one takes the restriction  of (\ref{coKpmgl}) to $U_q(sl_2) \otimes U_q(sl_2)$ and uses  (\ref{notK}). For the specialization $\alpha=1$, in addition one acts on the restriction of (\ref{coKpmgl}) with $\phi^{-1} \otimes \phi^{-1}$. 
Secondly, (\ref{cuKcepm}) produces the counit rules (\ref{cuc}), (\ref{cue}): this follows from (\ref{cuKpmgl}) and (\ref{notK}).
\end{proof}

In addition, the compatibility property satisfied by the antipode $S: U_q(sl_2) \rightarrow U_q(sl_2)$ holds:
\beqa
\mu \circ(id \otimes S) \circ \Delta(\cK^\pm) = \mu \circ(S \otimes id) \circ \Delta(\cK^\pm) = \eta \circ \varepsilon(\cK^\pm) \ \quad \mbox{for}\quad  \mbox{$\cK\equiv \cK_c$ or $\cK_e$}. \label{aKpmsl}
\eeqa
It follows from the restriction of (\ref{aKpmgl}) to $U_q(sl_2)$. Note that contrary to the FRT presentation with (\ref{ruleFRT}), the comultiplication rules  (\ref{DeltaKcpm}), (\ref{DeltaKepm}) have no group-like structure. So, the antipode of a Chevalley or equitable K-operator (i.e. $S(\cK_c^\pm)$ and $S(\cK_e^\pm)$) computed using (\ref{ac}), (\ref{ae}), doesn't take a simple form in terms of L and K-operators. \vspace{1mm}

To conclude this subsection, let us mention  that the central element of $U_q(sl_2)$ given by (\ref{Centc}) or equivalently (\ref{Cente}) can be derived from the Freidel-Maillet quantum determinant analog of (\ref{qdet2}). Namely, 
\beqa
{\rm det}_q^{FM} &=& {\rm tr}_{12}(U_{12} \cK_1^+ R^{(0)} \cK_2^-) =  - q^{-1}(q-q^{-1})^2\Omega_{c}  = - q^{-1}\Omega_{e}\ .\label{qdetFM}
\eeqa
Note that the central element  can be alternatively derived from the so-called quantum trace \cite[p. 46]{F1} with (\ref{Kopcbis}), (\ref{Kopebis}):
\beqa
{\rm tr}_q(\cK^+(\cK^-)^{-1}) = tr(D\cK^+(\cK^-)^{-1}) = (q-q^{-1})^2\Omega_{c}=\Omega_{e} \qquad \mbox{with} \qquad D = diag(q,q^{-1})\ .  
\eeqa

\subsection{K-operators of $U_q(sl_2)$ as intertwiners}\label{sec:intert}
The quasi-cocommutativity property of $U_q(sl_2)$ according to the Chevalley presentation is characterized by the universal R-matrix $\mathfrak{R} \in U_q(sl_2) \otimes U_q(sl_2)$ \cite{Dr0}. With respect to the comultiplication $\Delta$ and opposite comultiplication $\Delta'=\sigma \circ \Delta$, $\sigma(x \otimes y) = y \otimes x$, $\mathfrak{R}$  is the isomorphism such that $\Delta(x) = \mathfrak{R} \Delta'(x) \mathfrak{R}^{-1} $ for $x\in E,F,K^{\pm 1}$. It implies that the L-operators of $U'_q(sl_2)$ satisfy the intertwining relations 
 (see e.g. \cite{DEFK09}):
\beqa
(\rho_{V^{(1/2)}} \times id)\Delta(x) \cL^\pm|_{U'_q(sl_2)} =  \cL^\pm|_{U'_q(sl_2)} (\rho_{V^{(1/2)}} \otimes id)\Delta'(x) \quad \mbox{for any} \quad x\in K^{\pm 1/2}, E,F\ , 
\eeqa
as can be easily checked using (\ref{copc}),  (\ref{Lpm}) with (\ref{Kdemi}) and (\ref{repec}). In particular, R-matrices that solve the Yang-Baxter equation (\ref{YB0}) can be derived from $R_{12}=(\rho_1 \otimes \rho_2)(\mathfrak{R})$, where $\rho_i: U_q(sl_2) \rightarrow \End({\cal V}_i)$.  \vspace{1mm}

By analogy, it is expected that K-operators of $U_q(sl_2)$ satisfy certain intertwining relations.
Assume there exists a map $ \tilde\delta: U_q(sl_2) \rightarrow U_q(sl_2) \otimes U_q(sl_2)$ such that:
\beqa
(\rho_{V^{(1/2)}}\otimes id)\tilde\delta(K^{\pm 1})&=& \begin{pmatrix} q^{\pm 1}K^{\pm 1} & 0   \\
   0 & q^{\mp 1}K^{\pm 1} \end{pmatrix} \ ,\label{dc1}\\
(\rho_{V^{(1/2)}}\otimes id)\tilde\delta(E)&=& \begin{pmatrix} q^2E & K   \\
   0 & q^{-2}E \end{pmatrix}\ ,\quad (\rho_{V^{(1/2)}}\otimes id)\tilde\delta(F)= \begin{pmatrix} q^{-1}F & 0   \\
   K^{-2} & qF \end{pmatrix}\ ,\label{dc2}
\eeqa
for the Chevalley presentation, and
\beqa
&& (\rho_{V^{(1/2)}}\otimes id)\tilde\delta(X)= \begin{pmatrix} qX & 0   \\
   q-q^{-1} & q^{-1}X \end{pmatrix} \ ,\quad \ \ \ 
	(\rho_{V^{(1/2)}}\otimes id)\tilde\delta(X^{-1})= \begin{pmatrix} q^{-1}X^{-1} & 0   \\
   -(q-q^{-1})X^{-2} & qX^{-1} \end{pmatrix} \ ,\label{de1}\\
&& (\rho_{V^{(1/2)}}\otimes id)\tilde\delta(Y)= \begin{pmatrix} q^{-1}Y & 0   \\
0 & qY \end{pmatrix} \ ,\quad \quad \quad \quad 	
(\rho_{V^{(1/2)}}\otimes id)\tilde\delta(Z)= \begin{pmatrix} qZ & q^{-1}-q   \\
   0 & q^{-1}Z \end{pmatrix} \ ,\label{de2}
\eeqa
for the equitable presentation. By straightforward calculations, it is found that the defining relations of $U_q(sl_2)$ hold for the substitution $x \rightarrow (\rho_{V^{(1/2)}}\otimes id)\tilde\delta(x) \in  \End({\mathbb C}^2) \otimes U_q(sl_2)$. Furthermore, one finds that the following intertwining relations are satisfied by the K-operators (\ref{Kpmc}), (\ref{Kpme}):
\beqa
(\rho_{V^{(1/2)}} \otimes id)\tilde\delta(x) \cK_c^\pm =  \cK_c^\pm (\rho_{V^{(1/2)}} \otimes id)\Delta'(x) \quad \mbox{for} \quad x\in K^{\pm 1}, E,F\ , 
\eeqa
and
\beqa
(\rho_{V^{(1/2)}} \otimes id)\tilde\delta(x) \cK_e^\pm =  \cK_e^\pm (\rho_{V^{(1/2)}} \otimes id)\Delta'(x) \quad \mbox{for} \quad x\in X^{\pm 1},Y,Z\ , 
\eeqa
where (\ref{repec}), (\ref{repeq}) are used. \vspace{1mm}

The above results suggest to introduce a universal K-matrix for $U_q(sl_2)$, characterizing the relation between $U_q(sl_2)$-modules associated with $\Delta'$ and $\tilde\delta$. From this perspective, the universal K-matrix ${\mathfrak K}\in U_q(sl_2) \otimes U_q(sl_2)$ would be an invertible element such that $ \tilde\delta(x) = \mathfrak{K} \Delta'(x) \mathfrak{K}^{-1}$ for all $x \in U_q(sl_2)$. 
For the Chevalley presentation, this problem is solved via the embedding $U_q(sl_2) \hookrightarrow U'_q(sl_2)$. Indeed, observe that $\cK^\pm = (\rho_{V^{(1/2)}} \otimes id)({\cal F}) \cL^\pm|_{U'_q(sl_2)}$ where $\cal F = q^{(H\otimes H)/2}$ is a Drinfeld twist \cite{Dr0,R90} such that $K\equiv q^H$. Then, $\cK^\pm$ are obtained by specializing $\mathfrak{K}$, where $\mathfrak{K}$ gives an example of `modified' universal R-matrix \cite[Proposition 3.7]{Parm}. It is easily checked that the map $\tilde\delta$ with (\ref{dc1}), (\ref{dc2}), follows from the twisted comultiplication. For the equitable case, finding the universal K-matrix is an open problem that will be considered elsewhere. 
In any case, for finite dimensional tensor product representations K-matrices satisfying the Freidel-Maillet equation (\ref{FMeqint}) would follow from the universal K-matrix. Examples of K-matrices are given in the next subsection.\vspace{1mm}

To conclude this subsection, let us mention that similar results hold for the non-homogeneous Borel subalgebra ${\cal B}$ with defining relations (\ref{FMKKBorel}) and K-operator $\cK_{\cal B}$ in (\ref{Kborel}). Assume there are two maps $\tilde\delta_{\cal B}$, $\delta'_{\cal B}: \cal B \rightarrow U_q(sl_2) \otimes {\cal B}$ such that:
\beqa
(\rho_{V^{(1/2)}}\otimes id)\delta'_{\cal B}(Y)&=& \begin{pmatrix} qY & 0   \\
   0 & q^{-1}Y \end{pmatrix} \ ,\quad \quad \quad (\rho_{V^{(1/2)}}\otimes id)\delta'_{\cal B}(Z)= \begin{pmatrix} q^{-1}Z & 0   \\
   q-q^{-1} & qZ \end{pmatrix}  \ ,\nonumber\\
	(\rho_{V^{(1/2)}}\otimes id)\tilde\delta_{\cal B}(Y)&=& \begin{pmatrix} qY & 0   \\
   q-q^{-1} & q^{-1}Y \end{pmatrix}
	\ ,\quad 
	(\rho_{V^{(1/2)}}\otimes id)\tilde\delta_{\cal B}(Z)= \begin{pmatrix} q^{-1}Z & 0   \\
   0 & qZ \end{pmatrix} \ .\nonumber
\eeqa
It is straightforward to check that the image of the defining relation $qYZ-q^{-1}ZY = q-q^{-1}$  by above maps holds. Furthermore,
the following intertwining relations are obtained:
\beqa
(\rho_{V^{(1/2)}} \otimes id)\tilde\delta_{\cal B}(x) \cK_{\cal B} =  \cK_{\cal B} (\rho_{V^{(1/2)}} \otimes id)\delta'_{\cal B}(x) \quad \mbox{for} \quad x\in Y,Z\ .
\eeqa 
This suggests to investigate the universal version of (\ref{FMKKBorel}) and its solutions from the approach based on the universal analog of the intertwining relations above. 

\subsection{Constant K-matrices of $U_q(sl_2)$}
It is well-known that the simplest (non-trivial) constant solutions of the Yang-Baxter equation (\ref{YB0}) can be derived from  (\ref{Lpm}) by taking the fundamental representation of $U'_q(sl_2)$. Using (\ref{repec}) one gets 
\beqa
R= q^{1/2}\rho_{V^{(1/2)}}(\cL^+|_{U'_q(sl_2)}) \qquad \mbox{and} \qquad  R_{21}^{-1}= q^{-1/2}\rho_{V^{(1/2)}}(\cL^-|_{U'_q(sl_2)}) \ .\label{RfromL} 
\eeqa
By analogy, the simplest non-trivial examples of constant K-matrices of $U_q(sl_2)$ satisfying the Freidel-Maillet equation
\begin{align}
R_{12}\cK_{13}R^{(0)}_{12}\cK_{23}=\cK_{23}R^{(0)}_{12}\cK_{13}R_{12}\ \label{FMscal}
\end{align}
are derived from  (\ref{Kpmc}), (\ref{Kpme}) and (\ref{Kborel}). For the one-dimensional (spin $s=0$) representation, the images of the Chevalley K-operators are trivial. However,  the image of the equitable K-operators (\ref{Kpme}) or (\ref{Kborel}) is not, given by:
\beqa
\rho_{V^{(0)}}(\cK_e^\pm) =\rho_{V^{(0)}}(\cK_{X})= \rho_{V^{(0)}}(\cK_{\cal B})=\begin{pmatrix} 1 & 0  \\
   1 & 1 
      \end{pmatrix} \ .
\eeqa

For the spin $s=1/2$ representation, for the Chevalley case using (\ref{repec}) one gets  $\rho_{V^{(1/2)}}(\cK_c^+)=R^{(0)}R$ and $\rho_{V^{(1/2)}}(\cK_c^-)=qR^{(0)}R_{21}^{-1}$.
In this case, the K-matrices $\cK_c^\pm$ are just examples of constant solutions to braided Yang-Baxter equations.  For the equitable case, from (\ref{repeq}) one gets 
%
%
\begin{align}
\rho_{V^{(1/2)}}(\cK_e^+) =\left(
\begin{array}{cccc} 
 q    & 0 & 0 & 0 \\
0  &  q^{-1} & q^{-1}(q-q^{-1}) & 0 \\
1  &  0 & q^{-1} &  0 \\
0 & 1 & q-q^{-1} & q
\end{array} \right) \ ,\qquad \rho_{V^{(1/2)}}(\cK_e^-) =\left(
\begin{array}{cccc} 
 1    & 0 & 0 & 0 \\
0  &  1 & 0 & 0 \\
q^{-1}  &  -(q-q^{-1}) & 1 &  0 \\
0 & q & 0 & 1
\end{array} \right) \ \label{Kemat}
\end{align}
and
\begin{align}
\rho_{V^{(1/2)}}(\cK_{\cal B}) =\left(
\begin{array}{cccc} 
 q^{-1}    & 0 & q^{-2}-1 & q^{-2}-1 \\
q-q^{-1}  &  q & 1-q^{-2} & 1-q^{-2} \\
1  &  0 & q^{-1} &  q^{-1}-q \\
0 & 1 & 0 & q
\end{array} \right) \ ,\qquad \rho_{V^{(1/2)}}(\cK_{X}) =\left(
\begin{array}{cccc} 
 1    & 0 & 0 & 0 \\
0  &  1 & 0 & 0 \\
q  &  0 & 1 &  0\\
0 & q^{-1} & 0 & 1
\end{array} \right) \ .\label{Kborelmat}
\end{align}
As recalled in previous subsection,  the constant R-matrices (\ref{RfromL}) can be derived from the universal R-matrix of $U_q(sl_2)$, see e.g. \cite[Example 6.4.12]{CPb}. For the equitable case, it is expected that the constant K-matrices above can be derived from a universal K-matrix of equitable type.

\section{Spectral parameter dependent K-operators for $U_q(sl_2)$}\label{sec4}
In this section, we study a class of K-operators that solve the spectral parameter dependent Freidel-Maillet equation (\ref{RE}) recently revisited\footnote{For the class of solutions  associated with `twisted' Lax operators, (\ref{RE}) can be mapped to a braided Yang-Baxter equation. See e.g. \cite{Ku95}.} in \cite{Bas20}. We start from a K-operator (\ref{Kg}) with entries in   $   \cal A \otimes {\mathbb K}[u,u^{-1}]$, where $\cal A$ is a quotient of the alternating subalgebra of $U_q(\widehat{sl_2})$ introduced in \cite{Ter18}. See \cite[Section 4]{Bas20} for details. Certain specializations $\cal A \rightarrow U_q(sl_2)$ are considered. It is shown that corresponding images of the K-operator admit a decomposition in terms of $\cK^\pm$. 
Examples of spectral parameter dependent K-matrices are given. \vspace{1mm} 

As a preliminary, let us recall that the FRT presentation associated with $U'_q(sl_2)$ can be encoded into a single spectral parameter dependent Yang-Baxter equation \cite{F1}.
Introduce the quantum R-matrix  \cite{Baxter}
\begin{align}
R(u) =\left(
\begin{array}{cccc} 
 uq -  u^{-1}q^{-1}    & 0 & 0 & 0 \\
0  &  u -  u^{-1} & q-q^{-1} & 0 \\
0  &  q-q^{-1} & u -  u^{-1} &  0 \\
0 & 0 & 0 & uq -  u^{-1}q^{-1}
\end{array} \right) \ .\label{R}
\end{align}
It is known that $R(u)$ satisfies
%
%
the (spectral parameter dependent) Yang-Baxter equation  in the space ${\cal V}_1\otimes {\cal V}_2\otimes {\cal V}_3$: 
\begin{align}
R_{12}(u/v)R_{13}(u)R_{23}(v)=R_{23}(v)R_{13}(u)R_{12}(u/v)\ \qquad \forall u,v\ .\label{YB}
\end{align}
Note that the permutation operator  $P=R(1)/(q-q^{-1})$, and  $R_{12}(u)=PR_{12}(u)P=R_{21}(u)$. The FRT presentation
of $U'_q(sl_2)$ is encoded as
\beqa
K^{1/2}K^{-1/2}&=& K^{-1/2}K^{1/2}=1 \ ,\label{YBA00}\\
 R(u/v)\ (\cL(u)\otimes {\mathbb I})\ ( {\mathbb I} \otimes \cL(v))  &=&  ( {\mathbb I} \otimes \cL(v))\  (\cL(u)\otimes {\mathbb I})   \ R(u/v)  \ .\label{YBA}
\eeqa
Up to  automorphisms of $U'_q(sl_2)$, the so-called quantum  Lax operator  $(\cL(u))_{i,j}\in U'_q(sl_2)\otimes {\mathbb K}[u,u^{-1}]$ reads:
\beqa
&&  \cL(u)=
       \begin{pmatrix}
      u K^{1/2} -  u^{-1} K^{-1/2} & (q-q^{-1}) FK^{1/2}\\
  (q-q^{-1}) K^{-1/2}E &   u K^{-1/2} -  u^{-1} K^{1/2}
      \end{pmatrix} \ .\label{Lop}
\eeqa

In the form (\ref{YBA}), a generating function for central elements is given by the quantum determinant  ${\rm det}_q(\cL(u))$. For instance,   for $U'_q(sl_2)$ the Casimir element $\Omega_c$ given by (\ref{Centc}) follows from ${\rm det}_q(\cL(u))={\rm tr}_{12}\big(P^{-}_{12} \cL_1(u)\cL_2(uq))$ with $P^-_{12}=(1-P)/2$:
%
\beqa
{\rm det}_q(\cL(u)) = qu^2 + q^{-1}u^2 - (q-q^{-1})^2\Omega_c\ .
\eeqa

\vspace{1mm}

Our purpose is now to investigate the simplest spectral parameter dependent K-operator solutions of the Freidel-Maillet type equation (\ref{RE})
in light of the results in previous section.
To this end, we introduce a more general setting.\vspace{1mm}

Consider the $R-$matrices given by (\ref{R}) and (\ref{R0}). Note that $R^{(0)}$ can be viewed as a limiting case of (\ref{R}).
Introduce the associative algebra $\cal A$ with generators $\{\cW_0,\cW_1,\cZ_1,\tilde\cZ_1\}$ such that the following relations hold:
\beqa
\big[\cW_0,\cW_1\big]&=&k_+\tilde\cZ_1 - k_- \cZ_1\ ,\label{wg1}\\
\big[\cW_0,\cZ_1\big]_q&=& -k_+\bar\epsilon_+\ ,\qquad \big[\tilde\cZ_1,\cW_0\big]_q = -k_-\bar\epsilon_+\ ,\label{wg2}\\
\big[\cW_1,\tilde\cZ_1\big]_q&=& -k_-\bar\epsilon_-\ ,\qquad \big[\cZ_1,\cW_1\big]_q = -k_+\bar\epsilon_-\ ,\label{wg3} \\
\big[\cZ_1 , \tilde\cZ_1 \big] &=& (q-q^{-1})\big( \bar\epsilon_+ \cW_1 - \bar\epsilon_- \cW_0\big) \ ,\label{adwg}
\eeqa
where $k_\pm,\bar\epsilon_\pm$ are scalars in ${\mathbb K}(q)$. Note that the last relation (\ref{adwg}) can be derived from (\ref{wg1})-(\ref{wg3}), or alternatively the first relation (\ref{wg1}) can be derived from   (\ref{wg2})-(\ref{adwg}). Thus, in each case the defining relations for $\cal A$ correspond to a subset of (\ref{wg1})-(\ref{adwg}).
\begin{lem} 
\label{lem:Kg} The K-operator
\beqa
 \qquad\qquad \tilde \cK_g(u) &=&  \begin{pmatrix}  uq \cW_0 - u^{-1} \bar\epsilon_+ &  \cZ_1 + \frac{k_+ qu^2}{(q-q^{-1})}  \\ \tilde\cZ_1 +  \frac{k_- qu^2}{(q-q^{-1})}    
     &   uq \cW_1 - u^{-1}\bar\epsilon_-
      \end{pmatrix}\  \label{Kg}
\eeqa 
satisfies the spectral parameter dependent Freidel-Maillet equation
\begin{align} R(u/v)\ (\cK(u)\otimes {\mathbb I})\ R^{(0)}\ ({\mathbb I} \otimes \cK(v))\
= \ ({\mathbb I} \otimes \cK(v))\  R^{(0)}\ (\cK(u)\otimes {\mathbb I})\ R(u/v)\ 
\label{RE}\ .
 \end{align}
\end{lem}
%
%
%
%
%
%
%
%
%
\begin{proof} Insert (\ref{Kg}) in (\ref{RE}). One extracts the relations independent of $u$, which are given by (\ref{wg1})-(\ref{adwg}). 
\end{proof}
%
%
%
%
Note that (\ref{RE}) for the symmetric R-matrix (\ref{R}) can be viewed as a limiting case of a reflection equation \cite{Cher,Skly88}. 
For a Freidel-Maillet equation of the form (\ref{RE}), a generating function for central elements  is given by the so-called Sklyanin quantum determinant $\Gamma(u)$. By analogy with \cite[Proposition 5]{Skly88},  one finds 
\beqa
\Gamma(u)={\rm tr}_{12}\big(P^{-}_{12}(\cK(u)\otimes {\mathbb I})\ R^{(0)} ({\mathbb I} \otimes \cK(uq))\big) \ \ \label{gamma}
\eeqa
is  such that  $\big[\Gamma(u),(\cK(v))_{i,j}\big]=0$. A detailed proof can be found in \cite{Bas20}. For the K-operator $\tilde \cK_g(u)$, computing the r.h.s. of (\ref{gamma}) one gets:
\beqa
\Gamma(u) = u^2q^4 \Gamma_0  - (q-q^{-1})^2 q^2\Gamma_1 - k_+k_-u^4q^6 +  (q-q^{-1})^2\bar\epsilon_+\bar\epsilon_-u^{-2}\ ,
\eeqa  
where
\beqa
\Gamma_0&=& \frac{1}{2}\left(  (q-q^{-1})^2(\cW_0\cW_1 + \cW_1\cW_0)  - (q^2-q^{-2}) (k_+\tilde\cZ_1 + k_-\cZ_1) \right)\ ,\label{gam0}\\
\Gamma_1&=& \frac{1}{2}\left(  \cZ_1\tilde\cZ_1 + \tilde\cZ_1\cZ_1  + (q+q^{-1}) (\bar\epsilon_+\cW_1 + \bar\epsilon_-\cW_0) \right)\ ,\label{gam1}
\eeqa
are central in $\cal A$.\vspace{1mm}

Three different specializations $\varphi_c,\varphi_e,\varphi'_c: \cal A \rightarrow U_q(sl_2)$ are now considered. In the table below, the corresponding images of $\cW_0,\cW_1,\cZ_1,\tilde\cZ_1$, structure constants $\bar\epsilon_\pm,k_\pm$ and central elements $\Gamma_0,\Gamma_1$ are displayed.
%
%
%
%
%
%
\begin{table}[h]
\begin{tabular}{|c|c|c|c|}	
\hline 
${\cal A}$
 &  $\varphi_c({\cal A})$ & $\varphi_e({\cal A})$ & $\varphi'_c({\cal A})$  \\ \hline \hline
$\cW_0$
 &  $K$ & $X$ & $E$   \\ \hline
$\cW_1$
 & $K^{-1}$  & $Y$ &  $F$ \\ \hline
$\cZ_1$
 &  $(q-q^{-1})
FK$ & $
q^{-1}(YX-1)$ & $- 
(q-q^{-1})^{-1}K$   \\ \hline
$\tilde\cZ_1$
 &  $(q-q^{-1})
K^{-1}E$ & $-
 Z$ & $- 
(q-q^{-1})^{-1}K^{-1}$   \\ \hline
$\bar\epsilon_\pm$
 & $1$  & $1$ &  $0$ \\ \hline
$(k_+,k_-)$
 &  $(0,0)$ & $(0,
q-q^{-1})$ &  $(1,1
)$  \\ \hline
$\Gamma_0$
 & $(q-q^{-1})^2$  &  $(q-q^{-1})^2$ & $\Omega_c(q-q^{-1})^2$   \\ \hline
$\Gamma_1$
 & $\Omega_c(q-q^{-1})^2$  & $\Omega_e$ &  $(q-q^{-1})^{-2}$  \\ \hline
\end{tabular} \\
\vspace{2mm}
\caption{Three specializations ${\cal A} \rightarrow U_q(sl_2)$}
\label{table:spec}
\end{table}
Note that other specializations can be obtained applying $r: \ X \rightarrow Y$, $Y \rightarrow Z$,  $Z \rightarrow X$ to $\varphi_e(\cal A)$.\vspace{1mm}

The K-operator  (\ref{Kg}) is now specialized.  
From Lemma  \ref{lem:Kg}, Table \ref{table:spec}, 
it follows:
%
%
\beqa
 \cK_c(u)=
       \begin{pmatrix}
      uqK-u^{-1}  &  (q-q^{-1})FK \\
     (q-q^{-1})K^{-1}E & uqK^{-1} -u^{-1} 
      \end{pmatrix} \label{Kc} \ 
\eeqa
%
%
%
or
\beqa
\cK_e(u)=
       \begin{pmatrix}
      uqX - u^{-1}  &  q^{-1}(YX - 1)\\
     qu^2 - Z  & uqY - u^{-1} 
      \end{pmatrix} \ \label{Ke}
\eeqa
%
%
%
%
satisfy (\ref{RE}). Note that contrary to (\ref{Lop}), in both cases (\ref{Kc}), (\ref{Ke}), the entries 
\beqa
(\cK(u))_{i,j}\in U_q(sl_2)\otimes {\mathbb K}[u,u^{-1}]\ .
\eeqa
%

Other examples of equitable K-operators with spectral parameter can be obtained applying $r$ on (\ref{Ke}). For instance, denote $\cK_{{\cal B}X}(u) = r(K_e(u))$. By cyclicity of (\ref{eq:e4}) under the action of $r$, then
\beqa
\cK_{{\cal B}X}(u)=
       \begin{pmatrix}
      uqY - u^{-1}  &  q^{-1}(ZY - 1)\\
     qu^2 - X  & uqZ - u^{-1} 
      \end{pmatrix} \ \label{KBXu}
\eeqa
solves (\ref{RE}).\vspace{1mm}

The isomorphism (\ref{iso1}) can be derived using the connection between the Freidel-Maillet algebra (\ref{RE})  and the Yang-Baxter algebra (\ref{YBA}).  Adaptating \cite[Proposition 2]{Skly88}, \cite{FM91}, to the Freidel-Maillet type equation (\ref{RE}), let $\cK_0(u)$  be a solution of  (\ref{RE}). Assume there exists a pair of quantum Lax operators $\cL(u),\cL_0(u)$ satisfying  the exchange relations:
\beqa R(u/v)\ (\cL(u)\otimes {\mathbb I})\ ( {\mathbb I} \otimes \cL(v))  &=&  ( {\mathbb I} \otimes \cL(v))\  (\cL(u)\otimes {\mathbb I})   \ R(u/v)  \ ,  \label{YBA1}\\
 R(u/v)\ (\cL_0(u)\otimes {\mathbb I})\ ( {\mathbb I} \otimes \cL_0(v))  &=&  ( {\mathbb I} \otimes \cL_0(v))\  (\cL_0(u)\otimes {\mathbb I})   \ R(u/v)  \label{YBA2} \ ,\\
 R^{(0)}\ (\cL_0(u)\otimes {\mathbb I})\ ( {\mathbb I} \otimes \cL(v))  &=&  ( {\mathbb I} \otimes \cL(v))\  (\cL_0(u)\otimes {\mathbb I})   \ R^{(0)}  \label{YBA3} \qquad \forall u,v\  .
 \eeqa
Then, it is easy to show that the `dressed' K-operator defined by
\beqa
\tilde \cK(u) = \cL_0(u) \cK_0(u) \cL(u) \label{KL}
\eeqa
is also a solution of  (\ref{RE}). For instance, in addition to (\ref{Lop}) define  
\beqa
\cL_0(u)= diag(uK^{1/2},uK^{-1/2}) \label{L0}
\eeqa
and 
\beqa
 \cK_0^{e \rightarrow c}(u) =  \begin{pmatrix} u^{-1} & 0  \\
    q^{-1}  &  u^{-1} 
      \end{pmatrix}\ . \label{K0}
\eeqa
It is checked that $\cK_0^{e \rightarrow c}(u)$  satisfies (\ref{RE}). Also, by definition (\ref{YBA1})  holds. Then, the relations (\ref{YBA2})-(\ref{YBA3}) follow as limiting cases of  (\ref{YBA1}). By previous comments, 
\beqa
\cK_e(u) \rightarrow \cL_0(u) \cK_0^{e \rightarrow c}(u) \cL(u) \ \label{isoK}
\eeqa
relates (\ref{RE}) to (\ref{YBA1})-(\ref{YBA3}). Computing explicitely the r.h.s of (\ref{isoK}), one finds that the resulting expression coincides with the image of (\ref{Ke}) by the isomorphism (\ref{iso1}).

\subsection{Decomposition of K-operators}
The spectral parameter dependent equations (\ref{YBA00}), (\ref{YBA}), encode the defining relations of the FRT presentation of $U'_q(sl_2)$,  see \cite[p. 44]{F1} for details. Introduce the standard non-symmetric R-matrix $R^{ns}(u)$:
\beqa
 R^{ns}(u)= uR - u^{-1}R_{21}^{-1}\ . 
\eeqa
It is related with the symmetric R-matrix (\ref{R}) through the similarity transformation
\beqa
 R_{12}(u/v) &=& \cal M(u)_1  \cal M(v)_2 R^{ns}_{12}(u/v) \cal M(v)_2^{-1} \cal M(u)_1^{-1}\ ,\quad \mbox{with} \quad 
 {\cal M}(u)=   \begin{pmatrix}
         u^{1/2} & 0 \\
    0   &  u^{-1/2} 
                   \end{pmatrix} \ .\label{simil}
\eeqa
Defining $\cL^{ns}(u)=\cal M(u) \cL(u) \cal M(u)^{-1}$ with (\ref{Lop}), one has the decomposition $\cL^{ns}(u) = uq^{1/2} \cL^+ - u^{-1/2}q^{-1/2} \cL^-$. In terms of the new L-operators, the Yang-Baxter relation reads $R^{ns}(u/v)\cL^{ns}_1(u) \cL^{ns}_2(v)= \cL^{ns}_2(v)\cL^{ns}_1(u) R^{ns}(u/v)$. 
Extracting the independent relations in $R,R_{21}^{-1},\cL_1^\pm,\cL_2^\pm$ from the latter, one recovers the FRT presentation of $U'_q(sl_2)$ obtained from (\ref{frt1})-(\ref{frt3}).  \vspace{1mm}

A similar statement holds for the Freidel-Maillet equation (\ref{RE}) with specialized K-operators (\ref{Kc}) or (\ref{Ke}). 
To show that, we apply the same technique\footnote{In the context of non-ultralocal integrable models and braided Yang-Baxter algebras, this technique has been introduced in \cite{HK96}.}. For $\cK\in \{\cK_c,\cK_e\}$ with (\ref{Kc}), (\ref{Ke}),  define the new K-operators:
\beqa
\cK^{ns}(u)=\cal M(u) \cK(u) \cal M(u)^{-1}\ .
\eeqa
It is readily checked that they admit the following simple decomposition:
\beqa
\cK^{ns}(u) = uq\cK^+ - u^{-1} \cK^-\ .  
\eeqa
By (\ref{simil}), they satisfy the Freidel-Maillet equation for the non-symmetric R-matrix:
\begin{align} R^{ns}(u/v)\cK^{ns}_1(u) R^{(0)}\cK^{ns}_2(v)\
= \ \cK^{ns}_2(v) R^{(0)} \cK^{ns}_1(u) R^{ns}(u/v)\ 
\label{REns}\ .
 \end{align}
Expanding this equation in $u,v$, one extracts seven equations. Six coincide with (\ref{FMpp}), (\ref{FMpm}),  (\ref{RmKKpp}), (\ref{RmKmKp}), but  recall that (\ref{RmKKpp}), (\ref{RmKmKp}),  follow from the three equations (\ref{FMpp}), (\ref{FMpm}).
The remaining equation to show reads
\beqa
R\cK^-_1R^{(0)}\cK^+_2 - \cK^+_2R^{(0)}\cK^-_1R = R_{21}^{-1}\cK^+_1R^{(0)}\cK^-_2 - \cK^-_2R^{(0)}\cK^+_1R_{21}^{-1}\ .
\eeqa
Using $R-R_{21}^{-1}=(q-q^{-1})P$ and the Freidel-Maillet equations (\ref{FMpm}), (\ref{RmKmKp}), by simple computation one finds that this equation is satisfied. Thus, the independent relations are given by (\ref{FM1})-(\ref{FMpm}). We conclude that the Freidel-Maillet type equation (\ref{RE}) encodes the relations (\ref{FMpp}), (\ref{FMpm}) of Theorem \ref{thm:isoUslFM}. 

\subsection{K-matrices of $U_q(sl_2)$ with spectral parameter}
In the context of quantum integrable systems, R and K-matrices are the basic ingredient for the construction of mutually commuting quantities. For the class of quantum integrable systems generated from the Freidel-Maillet algebra of Theorem \ref{thm:isoUslFM}, in general for irreducible finite dimensional representations of $U_q(sl_2)$ the corresponding spectral parameter dependent Freidel-Maillet equation takes the form:
\begin{align}
R_{12}(u/v)\cK_{13}(u)R^{(0)}_{12}\cK_{23}(v)=\cK_{23}(v)R^{(0)}_{12}\cK_{13}(u)R_{12}(u/v)\ \in \ \mathrm{End}({\cal V}_1\otimes {\cal V}_2 \otimes {\cal V}_3).\label{FMscalu}
\end{align}
According to the choice of the Chevalley or equitable presentation of $U_q(sl_2)$, the simplest examples of K-matrices of $U_q(sl_2)$ with a spectral parameter  take a rather different form. For the one-dimensional (spin $s=0$) representation, the image of (\ref{Ke}) or (\ref{KBXu}) produces 
a non-trivial solution of (\ref{FMscalu}) given by:
\beqa
\rho_{V^{(0)}}(\cK_e(u)) = \rho_{V^{(0)}}(\cK_{{\cal B}X}(u))= \begin{pmatrix} uq-u^{-1} & 0  \\
   u^2q -1 & uq-u^{-1} 
      \end{pmatrix} \ .
\eeqa

For the spin $s=1/2$ representation, from  (\ref{Kc}), (\ref{Ke}) using (\ref{repec}), (\ref{repeq}), one gets: 
\begin{align}
&&\rho_{V^{(1/2)}}(\cK_c(u)) =\left(
\begin{array}{cccc} 
 uq^2-u^{-1}    & 0 & 0 & 0 \\
0  &  u-u^{-1} & q(q-q^{-1}) & 0 \\
0  &  q^{-1}(q-q^{-1}) & u-u^{-1} &  0 \\
0 & 0 & 0 & uq^2-u^{-1}
\end{array} \right) \ ,\label{Kcmatu}\\
&& \rho_{V^{(1/2)}}(\cK_e(u)) =\left(
\begin{array}{cccc} 
 uq^2-u^{-1}    & 0 & 0 & 0 \\
0  &  u-u^{-1} & q-q^{-1} & 0 \\
u^2q-q^{-1}  &  q-q^{-1} & u-u^{-1} &  0 \\
0 & q(u^2-1) & uq(q-q^{-1}) & uq^2-u^{-1}
\end{array} \right) \ .\label{Kematu}
\end{align}
From (\ref{KBXu}), one gets:
\begin{align}
&&\rho_{V^{(1/2)}}(\cK_{{\cal B}X}(u)) =\left(
\begin{array}{cccc} 
 u-u^{-1}    & 0 & q^{-1} -q & q^{-1} -q \\
u(q^2-1)  &  uq^2-u^{-1} & q-q^{-1} & q-q^{-1} \\
q(u^2-1)  &  0 & u-u^{-1} &  u(1-q^2) \\
0 & qu^2-q^{-1} & 0 & uq^2-u^{-1}
\end{array} \right) \ .\label{KBXmatu}
\end{align}
%

%
%



\vspace{1cm}

\noindent{\bf Acknowledgments:} I thank  Nicolas Cramp\'e for gratefully sharing a MAPLE code and comments, and Azat Gainutdinov for discussions. I thank Paul Terwilliger for comments on the manuscript and discussions, Hadewijch De Clercq and Bart Vlaar for communications and interest in this work. P.B.  is supported by C.N.R.S. 
\vspace{0.2cm}

\providecommand{\bysame}{\leavevmode\hbox to3em{\hrulefill}\thinspace}

\end{document}